\documentclass{article}       
\usepackage{amsmath,amssymb} 

\usepackage{graphicx}
\usepackage{amsmath,amssymb,amsthm} 

\usepackage[T1]{fontenc}
\usepackage[latin1]{inputenc}
\usepackage[scanall]{psfrag}
\usepackage{comment}
\usepackage{color}
\usepackage{setspace}

\newcommand{\DisplayNote}[1]{\hspace{.5cm}({\bf{#1}})}
\newcommand{\ResearchNote}[1]{\hspace{.5cm}}
\newcommand{\UnDisplayNote}[1]{}
\newcommand{\R}{{\mathbb{R}}}
\newcommand{\rplus}{{\mathbb{R}^{+}}}
\newcommand{\rpluscl}{\overline{\mathbb{R}}^{+}}

\newcommand{\C}{\mathbb{C}}
\newcommand{\cplus}{\overline{\mathbb{C}^{+}}}

\newcommand{\Z}{\mathbb{Z}}

\newcommand{\AB}     {{A\& B}}                  
\newcommand{\CDop}     {{C\& D}}                  

\newcommand{\Xscr} {{\mathcal X}}
\newcommand{\Yscr} {{\mathcal Y}}
\newcommand{\Uscr} {{\mathcal U}}
\newcommand{\Zscr} {{\mathcal Z}}

\newcommand{\bbm}[1]{\left[\begin{matrix} #1 \end{matrix}\right]}
\newcommand{\sbm}[1]{\left[\begin{smallmatrix} #1
             \end{smallmatrix}\right]}

\newcommand{\IOhat}{\mathbf{G}}

\newcommand{\SmallSysNode}{\sbm{\AB \cr \CDop}}

\newcommand{\abs}[1]{\left \vert #1 \right \vert}

\newcommand{\Dom}[1]{{\rm dom}\left (#1 \right )}

\newcommand{\Null}[1]{{\rm ker} \left (#1 \right )}

\newcommand{\rst}[1]{\big{|}_{#1}}

\newtheorem{theorem}{Theorem}[section]
\newtheorem{proposition}[theorem]{Proposition}

\newtheorem{definition}[theorem]{Definition}
\newtheorem{example}[theorem]{Example}

\newtheorem{remark}[theorem]{Remark}

\begin{document}

\title{Notes on glottal flow \\ and acoustic inertial effects}


\author{Jarmo~Malinen}



\date{11.11.2019}

\maketitle

\bibliographystyle{plain}

\begin{abstract}
  This text is a compilation of some of the notes that the author has
  written during the development of the low-order model ``DICO''
  \cite{A-A-M-M-V:MLBVFVTOPI,M-A-M-A-V:MLBVFOVTA,M-M:WPPGNVTR,M-M:IMBGSVTPG}
  for vowel phonation and the even more rudimentary glottal flow model
  \cite{M-A-M-G:PPMGFUIFHSV} for processing high-speed glottal video
  data.

  The following subject matters are covered: \rm{(i)}
  Incompressible, laminar, lossless flow models for idealised
  rectangular and wedge shape vocal fold geometries.  Equations of
  motion and the pressure distribution are computed in a closed form
  for each model using the unsteady Bernoulli's theorem; \rm{(ii)}
  The assumption of incompressibility and energy loss (i.e.,
  irrecoverable pressure drop) of the airflow in airways (including
  the glottis) is discussed using steady compressible Bernoulli
  theorem as the main tool; \rm{(iii)} Inertia of an uniform
  waveguide is studied in terms of the low-frequency limit of the the
  (acoustic) impedance transfer function. It is observed that the
  inductive loading in the boundary condition sums up with the
  waveguide inertance in an expected way; \rm{(iv)} It is shown
  that an acoustic waveguide, modelled by Webster's lossless equation
  with Dirichlet boundary condition at the far end, will produce the
  expected mass inertance of the fluid column as the low-frequency
  limit of the impedance transfer function.
\end{abstract}



\section{\label{IntroSec} Introduction}

This text is a cleaned-up compilation of notes that have been written
during the development of the low-order model ``DICO''
\cite{A-A-M-M-V:MLBVFVTOPI,M-A-M-A-V:MLBVFOVTA,M-M:WPPGNVTR,M-M:IMBGSVTPG}
for vowel phonation and the even more rudimentary glottal flow model
\cite{M-A-M-G:PPMGFUIFHSV} for processing high-speed glottal video
data.  While it is not possible to include all details and derivations
in journal articles, building even a modest model of phonation will
require a great number of mathematical and physical considerations,
idealisations, and approximations. There exists a number of fora (such
as arXiv.org) where complementary material can be presented almost
without any limitations, making excuses of obscurity somewhat moot
nowadays. Unfortunately, much of the background material of the above
mentioned publications is not included here due to shortness of time
and other resources of the author.  So much for the justification of
this text; let's move on.\footnote{The author has tried his best not
  to leave any obvious mistakes (mathematical, or of some other kind)
  in the material. If an interested reader makes observations, the
  author is pleased to receive comments by e-mail:
  jarmo.malinen@aalto.fi.}

Briefly, the subject of these notes is an air column inside a
perfectly smooth, acoustically reflecting tubular boundary. The air
column is both translating (when it is considered incompressible) as
well as it is an acoustic medium (where compressibility is a
prerequisite for the finite speed of sound).  The tube, i.e., the flow
channel consists of three parts having finite lengths: the
\emph{subglottal tract} (SGT), the \emph{glottis}, and the \emph{vocal
  tract} (VT). The much shorter glottis is positioned between SGT and
VT, and it is the only part of the tube boundary that is assumed to be
time-dependent. The walls of the tube at the glottis are called
\emph{vocal folds}. As is sometimes required, the SGT may be
considered as having been extended from its free end by a
\emph{piston} made of incompressible material (even, perhaps, fluid)
that has a much higher density than air. The mass and the dimensions
of the piston add to the total (flow-mechanical part of the) mass
inertia if they are taken into consideration. The acoustics is
considered only in the VT part using Webster's model, hence
restricting it to an immobile boundary.

An outline of these notes is as follows: In
Sections~\ref{RectangularSec}~and~\ref{WedgeSec}, incompressible,
laminar, lossless flow models are developed for two kinds of idealised
glottis geometries: \emph{rectangular} and \emph{wedge} shape vocal
folds.  Equations of motion and the pressure distribution are computed
in a closed form for each model using the unsteady Bernoulli's
theorem. This is motivated by, and only feasible due to the extremely
simplified nature of the glottis geometry.  From these equations, the
coefficient of \emph{inertance} shows up for each of the three parts
of the flow channel, and their sum -- the total inertance -- regulates
the mass inertial effects in the fluid movement.

In Section~\ref{CompressibleSec}, the assumption of incompressibility
is examined from the point of view of the flow. Constriction areas and
thermodynamic state are computed for glottal openings where the
compressible \emph{steady} flow would reach Mach 0.3 (often considered
as the upper limit for air flow to be treatable as incompressible) and
Mach 1.0. In Section~\ref{WaveGuideSec}, the VT inertance discovered
in the flow models of Sections~
\ref{RectangularSec}~and~\ref{WedgeSec} is associated to acoustics
using the uniform diameter acoustic waveguide as a model. It is, in
particular, observed that the inductive acoustic termination at the
waveguide end (mouth) will add to the inertance of the acoustic
system. That the same holds for general acoustic waveguides with
nonconstant intersectional areas is indicated in
Section~\ref{WebsterSec}.

The lossless, incompressible models introduced in
Sections~\ref{RectangularSec}~and~\ref{WedgeSec} are not as such
suitable for simulation modelling of vowels even when combined with VT
and SGT acoustics models. The two main reasons are the following:
\begin{enumerate}
  \item For performance reasons, it is desirable to have some
    \emph{pruning of terms} in the equations of motion (see
    Eqs.~\eqref{BoxGlottisEquationOfMotionVer1}~and~\eqref{VolumeVelocityEq})
    as well as in the equations of the hydrodynamic pressure
    components (see
    Eqs.~\eqref{BoxGlottisEquationOfPressureVer1}~and~\eqref{FormulaForPressureVolumeVer1})
    resulting in the aerodynamic force to the vocal folds (not
    considered in these notes). Pruning of a term is considered
    acceptable if its effect can be compensated by tuning of the model
    parameters in numerical simulations.
\item For accuracy reasons, some sort of (unrecoverable) pressure loss
  terms must be added.  Such pressure loss is due to fluid viscosity
  or various other entrance-exit effects not accounted for by the
  classical Hagen--Poiseuille's law.  Indeed, the pressure head lost
  due to friction cannot accelerate the column during the glottal open
  phase, and the flow at the glottal constriction is viscosity
  dominated right before the moment of closure. Incidentally, the
  glottal inertance is singular at the moment of closure as well.
\end{enumerate}



To conclude, the underlying story of these notes is the inertance that
makes appearance not only in the (incompressible) flow model for the
fluid column acceleration but in the equations of the acoustics of the
same fluid volume. All acoustic inertia in the proposed system is also
flow mechanical inertia but this is not true conversely: The inertia
of the piston (i.e., moving tissues during exhalation) as well as the
air jet separating from the lips are parts of the flow mechanical
loading, only.

The reader of these semi-informal notes should be warned that this
text is not, neither will it be, a proper scientific article. Many
necessary attributes of scientific articles are missing, including
much of the wider scientific context and all references to works of
other authors.




\section{\label{RectangularSec} Rectangular glottis}

In this section, we consider a rectangular glottis where an
incompressible, laminar\footnote{\ldots and what a convenient assumption this
  is!}, lossless flow takes place. The length of the SGT, glottis, and
VT denoted by $L_{SGT}, L_G, L_{VT}$, respectively.  The subintervals
$[-L_{SGT}, 0)$, $[0, L_G)$, and $[L_G, L_G + L_{VT}]$ denote these
    parts of the flow channel, and the flow intersection area
    $A(\cdot)$ is assumed to be time-dependent only on $[0, L_G)$.

\subsection{\label{ElementarySubSec} An elementary treatment}

The velocity, given the time-variant volume velocity $U = U(t)$
\begin{equation*}
  v(x,t) = 
  \begin{cases}
    & \frac{U(t)}{A} \quad \text{ for } x \in [-L_{SGT}, 0), \\
    & \frac{U(t)}{h g(t)} \quad \text{ for } x \in [0, L_{G}), \\
    & \frac{U(t)}{A} \quad \text{ for } x \in [L_{G}, L_{G} + L_{VT}). \\
  \end{cases}
\end{equation*}
The corresponding velocity potential, remembering its continuity, is given by
\begin{equation*}
  \psi(x,t) = 
  \begin{cases}
    & \frac{U(t)x}{A} \quad \text{ for } x \in [-L_{SGT}, 0), \\
    & \frac{U(t)x}{h g(t)} \quad \text{ for } x \in [0, L_{G}), \\
    & U(t)\left [\frac{x - L_G}{A} + \frac{L_G}{h g(t)} \right ] \quad \text{ for } x \in [L_{G}, L_{G} + L_{VT}). \\
  \end{cases}
\end{equation*}
The time derivative is given by
\begin{equation*}
\frac{\partial \psi}{ \partial t}(x,t) = 
  \begin{cases}
    & \frac{U'(t)x}{A} \quad \text{ for } x \in [-L_{SGT}, 0), \\
    & \frac{x}{h} \frac{d}{dt} \left( \frac{U(t)}{g(t)} \right ) \quad \text{ for } x \in [0, L_{G}), \\
    & \frac{x - L_G}{A} U'(t) + \frac{L_G}{h} \frac{d}{dt} \left( \frac{U(t)}{g(t)} \right)  \quad \text{ for } x \in [L_{G}, L_{G} + L_{VT}). \\
  \end{cases}
\end{equation*}
Unsteady Bernoulli at points $x = -L_{SGT}$ with pressure $p_s = p_s(t)$
and $x = L_G + L_{VT}$ with ambient pressure $p_{amb} = 0$:
\begin{equation*}
  -\frac{U'(t) L_{SGT}}{A} + \frac{1}{2} \left (\frac{U(t)}{A} \right )^2 + \frac{p_s(t)}{\rho}
  = \frac{L_{VT} + L_G - L_G}{A} U'(t) + \frac{L_G}{h} \frac{d}{dt} \left( \frac{U(t)}{g(t)} \right) + 
  \frac{1}{2} \left (\frac{U(t)}{A} \right )^2, 
\end{equation*}
or, equivalently, 
\begin{equation*}
 \frac{p_s(t)}{\rho}  = \frac{L_{SGT} + L_{VT}}{A} U'(t) + \frac{L_G}{h} \frac{d}{dt} \left( \frac{U(t)}{g(t)} \right). 
\end{equation*}
This is exactly the original version of the ``lossless model'' for the
rectangular glottis, given in \cite{M-A-M-G:PPMGFUIFHSV}.

\subsection{\label{ElegantSubSec} A more elegant version}

It is desirable to carry out the computation so that the inertance of
the full fluid column is treated in an unified manner. We show next
that, in fact,
\begin{equation} \label{EquationOfMotion}
 p_s(t) = \frac{d}{dt} \left( C_{iner}(t) U(t) \right)
\end{equation}
where 
\begin{equation} \label{RectangularInertance}
\begin{aligned}
  C_{iner}(t) & = \rho \int_{0}^{L_{SGT}} {\frac{ds}{A_{SGT}(s)}} + \rho \int_{0}^{L_{VT}} {\frac{ds}{A_{VT}(s)}} + \frac{\rho L_G}{h g(t)} \\
  & = C_{iner}^{(SGT)} + C_{iner}^{(VT)} + C_{iner}^{(G)}(t) \text{ with } C_{iner}^{(G)}(t) := \frac{\rho L_G}{h g(t)}
\end{aligned}
\end{equation}
is the \emph{total inertance} of the subglottal tract, vocal tract, and the interglottal volume.
Note that
\begin{equation} \label{InertanceDerivative}
  h C_{iner}'(t) g(t) = -\frac{\rho L_G g'(t)}{g(t)^2} g(t) = -\frac{\rho L_G g'(t)}{g(t)} = - \rho L_G \frac{d \ln{g(t)} }{dt}.
\end{equation}

Our assumptions are that the VT and SGT area functions $A_{VT}(s)$ and
$A_{SGT}(s)$ have their glottal ends at $s = L_G$, $s = 0$,
respectively, and that
\begin{equation*}
  \lim_{s \to L_{VT}}{A_{VT}(s)} =   \lim_{s \to L_{SGT}}{A_{SGT}(s)} = \infty
\end{equation*}
leading to stagnation at both of these ends.  For the incompressible
flow, the velocity is now given by
\begin{equation*}
  v(x,t) = 
  \begin{cases}
    & \frac{U(t)}{A_{SGT}(-x)} \quad \text{ for } x \in (-L_{SGT}, 0), \\
    & \frac{U(t)}{h g(t)} \quad \text{ for } x \in [0, L_{G}), \\
    & \frac{U(t)}{A_{VT}(x - L_G)} \quad \text{ for } x \in [L_{G}, L_{G} + L_{VT}). \\
  \end{cases}
\end{equation*}
The corresponding velocity potential, remembering its continuity, is given by
\begin{equation*}
  \psi(x,t) = 
  \begin{cases}
    & - U(t) \int_{x}^{0}\frac{ds}{A(-s)} \quad \text{ for } x \in [-L_{SGT}, 0), \\
    & \frac{U(t)x}{h g(t)} \quad \text{ for } x \in [0, L_{G}), \\
    & U(t) \left [\frac{L_G}{h g(t)} + \int_{L_G}^x \frac{ds}{A_{VT}(s - L_G)} \right ] \quad \text{ for } x \in [L_{G}, L_{G} + L_{VT}). \\
  \end{cases}
\end{equation*}
For $-L_{SGT} < x_1 < 0 < L_G < x_2 < L_{G} + L_{VT}$, the pressure drop satisfies by the unsteady Bernoulli equation
\begin{equation*}
  \frac{\partial \psi}{ \partial t}(x_1,t) + \frac{1}{2}\left (\frac{U(t)}{A_{SGT}(-x_1)} \right )^2 +  \frac{p(x_1,t)}{\rho}
=   \frac{\partial \psi}{ \partial t}(x_2,t) + \frac{1}{2}\left (\frac{U(t)}{A_{VT}(x_2)} \right )^2 +  \frac{p(x_2,t)}{\rho};
\end{equation*}
that is, 
\begin{equation*}
  p(x_1,t) - p(x_2,t)
  =   \rho \frac{\partial }{ \partial t} \left ( \psi (x_2,t) - \psi (x_1,t) \right ) 
  +  \frac{\rho}{2}\left (\frac{U(t)}{A_{VT}(x_2)} \right )^2  - \frac{\rho}{2}\left (\frac{U(t)}{A_{SGT}(-x_1)} \right )^2.
\end{equation*}
Denoting the stagnation pressures $p_s(t) = \lim_{x_1 \to - L_{SGT}}
p(x_1,t) $ and $p_{amb}(t) = \lim_{x_2 \to L_{SGT}} p(x_2,t)$ at the
infinitely wide ends of the tube, we get by taking the limits at the
both ends
\begin{equation*}
\begin{aligned}
  p_s(t) - p_{amb}(t)
  & =   \rho \frac{\partial }{ \partial t} \left ( 
U(t) \left [\frac{L_G}{h g(t)} + \int_{L_G}^{L_G + L_{VT}} \frac{ds}{A_{VT}(s - L_G)} \right ]
+ U(t) \int_{-L_{SGT}}^{0}\frac{ds}{A(-s)} \right ) \\
& =  \frac{\partial }{ \partial t} \left ( U(t) C_{iner}(t) \right )
\end{aligned}
\end{equation*}
by change of variables. Assuming that $p_{amb}(t) = 0$ is a constant
reference ambient pressure level, the result
Eq.~\eqref{EquationOfMotion} follows.

\begin{remark}
For later comparison, Eq.~\eqref{EquationOfMotion} can we written in
terms of the velocity at the glottal opening $v(t) = U(t)/hg(t)$ as
\begin{equation*} 
\begin{aligned}
 v'(t)  & =   \frac{1}{C_{iner}(t) h g(t)} \left ( p_s(t) - h \left ( C_{iner}(t) g'(t) + C_{iner}'(t) g(t) \right ) v(t) \right ) \\
 & = \frac{1}{C_{iner}(t) h g(t)} \left ( p_s(t) - h g'(t) \left (C_{iner}(t)  - \frac{\rho L_G}{h g(t)} \right ) v(t) \right )
\end{aligned}
\end{equation*}
Thus we get
\begin{equation} \label{BoxGlottisEquationOfMotionVer1}
\begin{aligned}
 v'(t) & = \frac{1}{ C_{iner}(t)h g(t)} \left ( p_s(t) - h g'(t) C_{iner}^{(TOT)} v(t) \right )
\end{aligned}
\end{equation}
where $C_{iner}^{(TOT)} := C_{iner}^{(SGT)} + C_{iner}^{(VT)}$ is the
total inertance \emph{excluding} the time-dependent glottis.
\end{remark}

\subsection{\label{RectangularSolutionSubSec} Solving the glottal flow and pressure}
By integration of Eq.~\eqref{EquationOfMotion},
\begin{equation*}
C_{iner}(t) U(t)  - C_{iner}(0) U(0) = \int_0^t {p_s(\tau) \, d \tau}.
\end{equation*}
Assuming that $U(0) = 0$ (which is reasonable by fixing the opening
point in time) we get $U(t) = \frac{1}{C_{iner}(t)} \int_0^t
{p_s(\tau) \, d \tau}$.  For the velocity in the rectangular glottal
channel we get
\begin{equation} \label{VelocityGlottis}
  v(t) = \frac{1}{C_{iner}(t) h g(t)} \int_0^t {p_s(\tau) \, d \tau}.
\end{equation}
Using again the unsteady Bernoulli at $- L_{SGT} < x_1 < 0$ and $x \in [0, L_G]$,  we get
\begin{equation*}
  p(x_1, t)  - p(x, t) 
  =   \rho \frac{\partial }{ \partial t} \left ( \psi \left (x,t \right ) - \psi (x_1,t) \right ) 
  +  \frac{\rho}{2} v(t)^2 - \frac{\rho}{2}\left (\frac{U(t)}{A_{SGT}(x_1)} \right )^2.  
\end{equation*}
Taking the limit $x_1 \to - L_{SGT}$ and noting the stagnation to pressure $p_s(t)$, we get
\begin{equation*}
  p_s(t)  - p(x, t) 
  =   \frac{\partial}{\partial t} U(t)  \left ( \frac{\rho  x}{h g(t)} 
  +  \rho  \int_{-L_{SGT}}^{0}{\frac{ds}{A(-s)}} \right ) 
  +  \frac{\rho}{2} v(t)^2.
\end{equation*}
Thus, the pressure in the glottis $x \in [0, L_G]$ is given by
\begin{equation*}
\begin{aligned}
   p(x, t) 
   & = p_s(t)  -  \frac{\rho}{2} v(t)^2 
   -  \frac{\partial }{ \partial t} U(t)  \left ( \frac{\rho x }{h g(t)} +  C_{iner}^{(SGT)}\right ) \\
   & = p_s(t)  -  \frac{\rho}{2} v(t)^2 - C_{iner}^{(SGT)} U'(t)
   -  \rho x \frac{\partial }{ \partial t}  \left ( \frac{U(t) }{h g(t)}   \right ) \\
   & = p_s(t)  -  \frac{\rho}{2} v(t)^2  -  \rho x  v'(t)  - C_{iner}^{(SGT)} U'(t).
\end{aligned}
\end{equation*}
Writing $U'(t) = h \left (g'(t) v(t) + g(t)v'(t) \right )$ yields the first version of the equations for the pressure:
\begin{equation} \label{BoxGlottisEquationOfPressureVer1}
\begin{aligned}
   p(x, t) & = p_s(t) - \frac{\rho}{2} v(t)^2 - C_{iner}^{(SGT)} U'(t)
   - \rho x \frac{\partial }{ \partial t} \left ( \frac{U(t) }{h g(t)}
   \right ) \\ & = p_s(t) - \frac{\rho}{2} v(t)^2 - \left (\rho x + h
   C_{iner}^{(SGT)} g(t) \right ) v'(t) - h C_{iner}^{(SGT)} g'(t)
   v(t).
\end{aligned}
\end{equation}
The third term on the RHS relates to the increase of pressure at the
narrowing due to deceleration (i.e., when $v'(t) < 0$).

Now, it is possible to eliminate the $v'(t)$ term from
Eq.~\eqref{BoxGlottisEquationOfPressureVer1}.  Inserting
Eq.~\eqref{BoxGlottisEquationOfMotionVer1} gives for $x \in [0, L_G)$
\begin{equation} \label{BoxGlottisEquationOfPressureVer2}
\begin{aligned}
   p(x, t) & = p_s(t) - \frac{\rho}{2} v(t)^2  - h C_{iner}^{(SGT)} g'(t) v(t) \\
   & - \left (\rho x + h C_{iner}^{(SGT)} g(t) \right ) \cdot \frac{1}{h C_{iner}(t) g(t)} \left ( p_s(t) - h g'(t) C_{iner}^{(TOT)} v(t) \right ) \\
   & = \left [1 - \frac{\rho x + h C_{iner}^{(SGT)} g(t)}{h C_{iner}(t) g(t)} \right ] p_s(t) - \frac{\rho}{2} v(t)^2   \\
   & + \left [  \frac{\left (\rho x + h C_{iner}^{(SGT)} g(t) \right ) \left ( h g'(t) C_{iner}^{(TOT)} \right )}{h C_{iner}(t) g(t)} - h C_{iner}^{(SGT)} g'(t)  \right ]  v(t) \\
   & = \left [1 - \frac{C_{iner}^{(SGT)} + \frac{\rho x}{h g(t)} }{ C_{iner}(t)} \right ] p_s(t) - \frac{\rho}{2} v(t)^2   \\ 
   & \quad \quad \quad \quad + \frac{\rho  g'(t)}{C_{iner}(t) g(t)} \left [C_{iner}^{(TOT)} x -  C_{iner}^{(SGT)} L_G     \right ]  v(t).
\end{aligned}
\end{equation}

A few concluding remarks are now in order.  The first two terms on the
RHS of the top row in Eq.~\eqref{BoxGlottisEquationOfPressureVer2} are
familiar from the steady Bernoulli principle, and one should note that
always $\frac{C_{iner}^{(SGT)} + \frac{\rho x}{h g(t)}}{C_{iner}(t)} <
1$ in the second to the last row.  One could call this number
\emph{partial inertance proportion} for an obvious reason, and that
correction remains there even in a stationary glottis.  It is unclear
to me what the last term in
Eq.~\eqref{BoxGlottisEquationOfPressureVer2} stands for but it
certainly vanishes for nonmoving glottis. The expression in the
brackets can be written as $C_{iner}^{(SGT)} (x - L_g) +
C_{iner}^{(VT)}x$.

\ResearchNote{JM: Why does it not reduce to the ordinary steady
  incompressible Bernoulli if $g'= p_s' =0$? Well, without losses this
  would amount to an accelerating flow if $p_s > 0$, not steady unless
  $p_s = 0$. Make a remark of this.}


\section{\label{WedgeSec} Wedge-like glottis}

Let us carry out similar computations as in
Section~\ref{RectangularSec} but this time for wedge-like vocal
folds. We again assume that the length of the glottis is $L_G$, and
the glottal part of the airways is the interval $[0, L_G]$.  The
smallest opening is denoted by $g = g(t) > 0$, and the opening at the
wide end is denoted by $g_0$. The narrow end is always downstream.

The glottal area function is now given by
\begin{equation*}
  A_G(x) = h \left (\frac{g(t)- g_0}{L_G} x + g_0 \right ) \text{ for } x \in [0, L_G],
\end{equation*}
giving for the glottal inertance
\begin{equation} \label{GlottalIntertanceForWedges}
  C_{iner}^{(G)}(t) = \rho \int_0^{L_G} {\frac{ds}{ A_G(s)}} =
  \frac{\rho L_G}{h(g_0 -  g(t))} \ln{\left (\frac{g_0}{g(t)} \right )}
\end{equation}
where we used the fact that 
\begin{equation*}
\begin{aligned}
   \int_{0}^x {\frac{d s}{(g(t)- g_0) s + g_0 L_G}} & = \frac{1}{g(t)-
     g_0} \ln{\frac{(g(t)- g_0) x + g_0 L_G}{g_0 L_G}} \\ & =
   \frac{1}{g(t)- g_0} \ln{\left (1 - \frac{g_0 - g(t)}{g_0 L_G} x
     \right )}.
\end{aligned}
\end{equation*}
Note that the wedge geometry gives a logarithmic singularity in $g(t)$
for $C_{iner}^{(G)}(t)$ at $g(t) = 0$. Compared to
Section~\ref{RectangularSec}, the singularity is there stronger since
$C_{iner}^{(G)}(t) = L_G/h g(t)$ for the rectangular glottis.  In any
case, it is a general fact in any geometry that the inertance of the
glottis (and hence, the total inertance of the fluid column) becomes
singular at the moment of closure.

\subsection{Equation of motion for wedge-like vocal folds}

Let us start, again, from the equation of motion for the fluid column.
For the incompressible, lossless flow, the velocity is
\begin{equation*}
  v(x,t) = 
  \begin{cases}
    & \frac{U(t)}{A_{SGT}(-x)} \quad \text{ for } x \in (-L_{SGT}, 0), \\
    & \frac{U(t) L_G}{h \left ( (g(t)- g_0) x + g_0 L_G \right )} \quad \text{ for } x \in [0, L_{G}), \\
    & \frac{U(t)}{A_{VT}(x - L_G)} \quad \text{ for } x \in [L_{G}, L_{G} + L_{VT}). \\
  \end{cases}
\end{equation*}
The corresponding velocity potential, remembering its continuity, is given by
\begin{equation*}
  \psi(x,t) = 
  \begin{cases}
    & - U(t) \int_{x}^{0}{\frac{ds}{A(-s)}} \quad \text{ for } x \in [-L_{SGT}, 0), \\
    &  \frac{L_G}{h}  \frac{U(t)}{g(t)- g_0} 
        \ln{\left (1 + \frac{g(t)- g_0}{g_0 L_G} x \right )} \quad \text{ for } x \in [0, L_{G}), \\
    & U(t) \left [\rho^{-1} C_{iner}^{(G)}(t)  
 + \int_{L_G}^x \frac{ds}{A_{VT}(s - L_G)} \right ] \quad \text{ for } x \in [L_{G}, L_{G} + L_{VT}). \\
  \end{cases}
\end{equation*}
For $-L_{SGT} < x_1 < 0 < L_G < x_2 < L_{G} + L_{VT}$, the pressure drop satisfies by the unsteady Bernoulli equation
\begin{equation*}
  \frac{\partial \psi}{ \partial t}(x_1,t) + \frac{1}{2}\left (\frac{U(t)}{A_{SGT}(-x_1)} \right )^2 +  \frac{p(x_1,t)}{\rho}
=   \frac{\partial \psi}{ \partial t}(x_2,t) + \frac{1}{2}\left (\frac{U(t)}{A_{VT}(x_2)} \right )^2 +  \frac{p(x_2,t)}{\rho};
\end{equation*}
that is, 
\begin{equation*}
  p(x_1,t) - p(x_2,t)
  =   \rho \frac{\partial }{ \partial t} \left ( \psi (x_2,t) - \psi (x_1,t) \right ) 
  +  \frac{\rho}{2}\left (\frac{U(t)}{A_{VT}(x_2)} \right )^2  - \frac{\rho}{2}\left (\frac{U(t)}{A_{SGT}(-x_1)} \right )^2.
\end{equation*}
Denoting the stagnation pressures $p_s(t) = \lim_{x_1 \to - L_{SGT}}
p(x_1,t) $ and $p_{amb}(t) = \lim_{x_2 \to L_{SGT}} p(x_2,t)$ at the
infinitely wide ends of the tube, we get by taking the limits at the
both ends
\begin{equation*}
\begin{aligned}
  p_s(t) - p_{amb}(t) & = \rho \frac{\partial }{ \partial t} \left (
  U(t) \left [\rho^{-1} C_{iner}^{(G)}(t) + \int_{L_G}^{L_G + L_{VT}}
    \frac{ds}{A_{VT}(s - L_G)} \right ] + U(t)
  \int_{-L_{SGT}}^{0}{\frac{ds}{A(-s)}} \right ) \\ & = \frac{d
    }{ d t} \left ( U(t) C_{iner}(t) \right )
\end{aligned}
\end{equation*}
by change of variables. Assuming again that $p_{amb}(t) = 0$ is a
constant reference ambient pressure level, the equation of motion
\eqref{EquationOfMotion} follows by defining the total inertance by
$C_{iner}(t) = C_{iner}^{(SGT)} + C_{iner}^{(VT)} +
C_{iner}^{(G)}(t)$. Obviously
\begin{equation*}
  U'(t)  = \frac{1}{C_{iner}(t) } \left ( p_s(t) - C_{iner}'(t) U(t) \right ) 
\end{equation*}
where $C_{iner}(t) = C_{iner}^{(TOT)} + C_{iner}^{(G)}(t)$ and 
\begin{equation} \label{InertanceDerivative}
\begin{aligned}
  C_{iner}^{(G)}\prime(t) & = \frac{\rho L_G}{h} \left (\frac{- g'(t)}{(g_0 -  g(t))^2} \ln{\left (\frac{g_0}{g(t)} \right )} + \frac{1}{g_0 -  g(t)} \frac{-g'(t)}{g(t)} \right ) \\
& = - \frac{g'(t)}{g_0 -  g(t)}  \left (  C_{iner}^{(G)}(t) + \frac{\rho L_G}{h g(t)} \right ).
\end{aligned}
\end{equation}
These two Eqs. together give the differential equation for the volume velocity that we
will use for expressing the glottal pressure distribution:
\begin{equation} \label{VolumeVelocityEq}
\begin{aligned}
  U'(t)  & = \frac{1}{C_{iner}(t) } \left ( p_s(t) + \frac{g'(t)}{g_0 -  g(t)}  \left (  C_{iner}^{(G)}(t) + \frac{\rho L_G}{h g(t)} \right ) U(t) \right ) \\
  & = \frac{1}{C_{iner}(t) } \left ( p_s(t) + \frac{\rho L_G g'(t)}{h ( g_0 -  g(t) )^2}  
  \left ( \frac{g_0}{g(t)} +  \ln{\left (\frac{g_0}{g(t)} \right )}  - 1\right ) U(t) \right ) \\
  & = \frac{1}{C_{iner}(t) } \left ( p_s(t) + \frac{\rho L_G g'(t)}{h ( g_0 -  g(t) )^2}  
  \left ( \frac{g_0}{g(t)} +  \ln{\left (\frac{g_0}{e g(t)} \right )} \right ) U(t) \right ).
\end{aligned}
\end{equation}
Recalling $g(t) \ll g_0$ and omitting the weaker logarithmic singularity, we get the approximation
\begin{equation} \label{VolumeVelocityApproxEq}
\begin{aligned}
  U'(t)  &  = \frac{1}{C_{iner}(t) } \left ( p_s(t) + \frac{\rho L_G g'(t)}{h  g_0 g(t)} U(t) \right ) \\
  & = \frac{1}{C_{iner}(t) } \left ( p_s(t) + \frac{\rho L_G}{h  g_0} \cdot \frac{d \ln{g(t)}}{dt} U(t) \right ).
\end{aligned}
\end{equation}
which is quite elegant.

\subsection{Equation for the flow velocity}

Instead of the equations \eqref{VolumeVelocityEq} and
\eqref{VolumeVelocityApproxEq} for the volume velocity $U(t)$, let us
turn to the flow velocity at the glottal opening
\begin{equation*}
  v(t) = \frac{U(t)}{h g(t)} 
\end{equation*}
which satisfies a much more involved differential equation.  In
explicit terms, the wedge glottis without losses gives
\begin{equation*}
\begin{aligned}
  p_s(t) & = C_{iner}(t) \left ( v'(t) h g(t) + v(t) h g'(t) \right )
  + C_{iner}^{(G)}\prime (t) v(t) h g(t) \\
  & =  C_{iner}(t) h g(t) \cdot  v'(t) + \left (C_{iner}^{(G)} \prime (t)  h g(t) +  C_{iner}(t) h  g'(t)  \right ) v(t)
\end{aligned}
\end{equation*}
and solved for the acceleration at the glottis, it gives
\begin{equation*} 
   v'(t)  =  \frac{1}{C_{iner}(t) h g(t)} \left (  p_s(t) - h  \left (C_{iner}^{(G)} \prime (t) g(t) +  C_{iner}(t)   g'(t)  \right ) v(t) \right ).
\end{equation*}
This together with Eq.~\eqref{InertanceDerivative} yield the equation
for motion
\begin{equation} \label{WedgeGlottisEqsMotion}
     v'(t)  =  \frac{1}{C_{iner}(t) h g(t)} \left (  p_s(t) + h g'(t)  \left ( \frac{g(t)}{g_0 -  g(t)} C_{iner}^{(G)}(t) -   C_{iner}(t)  + \frac{\rho L_G}{h(g_0 -  g(t))}   \right ) v(t) \right )
\end{equation}
where $C_{iner}^{(G)}(t)$ is given by
\eqref{GlottalIntertanceForWedges}, $C_{iner}(t) = C_{iner}^{(G)}(t) +
C_{iner}^{(TOT)}$ where $C_{iner}^{(TOT)} = C_{iner}^{(SGT)} +
C_{iner}^{(VT)}$ is the total inertance of the nonmoving part of the
vocal tract.

 Note that the first two terms inside the parentheses can
 be joined as
 \begin{equation*}
 \begin{aligned}
 &  C_{iner}^{(G)}(t)\left ( \frac{g(t)}{g_0 -  g(t)}    - 1 \right ) - C_{iner}^{(TOT)}  
  =    - \frac{g_0 - 2 g(t)}{g_0 -  g(t)}C_{iner}^{(G)}(t)  - C_{iner}^{(TOT)} 
 \end{aligned}
 \end{equation*}
 which gives the final form of the unsimplified equations of the
 motion
 \begin{equation}  \label{WedgeGlottisEqsMotion}
\begin{aligned} 
      v'(t)  & =  \frac{1}{C_{iner}(t) h g(t)} \cdot \\ 
      & \cdot \left (  p_s(t) - g'(t)  \left [ h C_{iner}^{(TOT)} 
      + \frac{\rho L_G}{g_0 -  g(t)} \left ( \frac{g_0 - 2 g(t)}{g_0 -  g(t)}  \ln{\frac{g_0}{g(t)}} 
      -  1 \right ) \right ] v(t) \right ) \\
      & =  \frac{1}{C_{iner}(t) h g(t)} \cdot \\ 
      & \cdot \left (  p_s(t) - h g'(t)  \left [C_{iner}^{(TOT)} 
      + \frac{g_0 - 2 g(t) }{g_0 -  g(t)} \left ( C_{iner}^{(G)}(t) -  \frac{\rho L_G}{h(g_0 - 2 g(t))} \right ) \right ] v(t) \right ).
\end{aligned}
 \end{equation}

\begin{remark}
  Note that the expression $C_{iner}(t) h g(t)$ in the denominator of
  Eq.~\eqref{WedgeGlottisEqsMotion} has a removable singularity as
  $g(t) \to 0$ in the sense that the limit at the closing moment $=
  0$. See Eq.~\eqref{GlottalIntertanceForWedges}.
\end{remark}


\subsubsection*{Simplifications based on the wedge geometry}

If $g(t) \ll g_0$ as is usually the case in the wedge geometry, we get
\begin{equation} \label{WedgeGlottisEqsMotionSimplifiedVer1}
     v'(t)  =  \frac{1}{C_{iner}(t) h g(t)} 
\left (  p_s(t) - h g'(t)  \left [ C_{iner}(t) - \frac{\rho L_G}{ h g_0}  \right ] v(t) \right )
\end{equation}
which is directly comparable with the corresponding formula
Eq.~\eqref{BoxGlottisEquationOfMotionVer1} for the rectangular
glottis.  Not that the sum in brackets is always nonnegative, and it
grows unboundedly as $g(t) \to 0$ right before the closure. So, at the
closing glottis when $g'(t) < 0$, the effect of the additional term is
to work in the same direction as the stagnation pressure $p_s(t)$.

Note that $\frac{ \rho L_G}{h g_0}$ in
Eq.~\eqref{WedgeGlottisEqsMotionSimplifiedVer1} is the inertance of a
tube of length $L_G$ with uniform intersectional area $h g_0$. Since
$g_0 \gg e g(t)$, that term is insignificant compared to
$C_{iner}^{(G)}$ as well as to $C_{iner}^{(TOT)}$, and we get an even
more simplified model
\begin{equation} \label{WedgeGlottisEqsMotionSimplifiedVer2}
     v'(t)  =  \frac{1}{C_{iner}(t) h g(t)}  \left (  p_s(t) -    h g'(t) C_{iner}(t) v(t) \right ).
\end{equation}

The further approximation for
Eqs.~\eqref{WedgeGlottisEqsMotionSimplifiedVer2}--\eqref{WedgeGlottisEqsMotionSimplifiedVer2}
is to replace $C_{iner}(t)$ given by Eq.~\eqref{GlottalIntertanceForWedges}
by the simplified form (again by  $g(t) \ll g_0$)
\begin{equation*}
  C_{iner}(t) = C_{iner}^{(TOT)} + \frac{\rho L_G}{h g_0} \ln{\left (\frac{g_0}{g(t)} \right )}
\text{ or even } C_{iner}(t) = C_{iner}^{(TOT)}.
\end{equation*}
One such variant produces from Eq.~\eqref{WedgeGlottisEqsMotionSimplifiedVer1}
the form
\begin{equation} \label{WedgeGlottisEqsMotionSimplifiedVer3}
     v'(t) = \frac{1}{C_{iner}^{(TOT)} h g(t)} \left ( p_s(t) - h
     g'(t) \left [ C_{iner}^{(TOT)} + \frac{\rho L_G}{h g_0} \ln{\left
         (\frac{g_0}{e g(t)} \right )} \right ] v(t) \right )
\end{equation}
where the logarithmic singularity of the glottal inertance is still
present in one place.

Finally, removing the glottal opening velocity $g'(t)$, we get back to
the lossless wedge model in DICO
\cite{A-A-M-M-V:MLBVFVTOPI,M-A-M-A-V:MLBVFOVTA}. Note that here the
glottal gap $g(t)$ is in the denominator in
Eqs.~\eqref{WedgeGlottisEqsMotion}--\eqref{WedgeGlottisEqsMotionSimplifiedVer3}
since we did not extend the glottis downstream to a control surface of
constant area, right above the glottis.

\begin{remark}
Within the limits of the approximation $g(t) \ll g_0$,
Eq.~\eqref{WedgeGlottisEqsMotionSimplifiedVer2} can plainly written in
the form
\begin{equation*} 
     v'(t)  =  \frac{1}{C_{iner}(t) h g(t)} 
\left (  p_s(t) -  C_{iner}(t) h g'(t)  v(t) \right ); 
\end{equation*}
that is, 
\begin{equation*} 
  p_s(t)  =  C_{iner}(t) h \frac{d g(t)v(t)}{dt} = C_{iner}(t) U'(t). 
\end{equation*}
This differs from the original equation of motion only by the
approximation $C_{iner}'(t) = C_{iner}^{(G)}\prime(t) = 0$. What the
above reasoning amounts to, is just showing in what approximative
sense $C_{iner}^{(G)}(t)$ can be regarded as being constant of time.
If we do not want to make that rather crude approximation, we should
use Eq.~\eqref{WedgeGlottisEqsMotionSimplifiedVer1} instead of
Eq.~\eqref{WedgeGlottisEqsMotionSimplifiedVer2}. In fact,
Eq.~\eqref{WedgeGlottisEqsMotionSimplifiedVer1} is equivalent with
\begin{equation} \label{WedgeGlottisEqsMotionSimplifiedVer1Equivalent}
  p_s(t)  = C_{iner}(t) U'(t) - \frac{\rho L_G g'(t)}{h g_0 g(t)} U(t); 
\end{equation}
that is, by the approximation $C_{iner}^{(G)}\prime(t) = -\frac{\rho
  L_G g'(t)}{h g_0 g(t)}$ where only one term is omitted from
Eq.~\eqref{InertanceDerivative}.
\end{remark}

\subsection{Glottal pressure for the wedge-like vocal folds}

As above for the rectangular glottis, we get from the unsteady
Bernoulli principle
\begin{equation*}
\begin{aligned}
  & p_s(t)  - p(x, t) \\
  &  =   \frac{\partial}{\partial t} U(t)  \left ( - \frac{\rho L_G}{h(g_0 - g(t))} 
        \ln{\left (1 - \frac{g_0 - g(t)}{g_0 L_G} x \right )}
  +  \rho  \int_{-L_{SGT}}^{0}{\frac{ds}{A(-s)}} \right ) 
  +  \frac{\rho}{2} v(x,t)^2
\end{aligned}
\end{equation*}
where 
\begin{equation*}
  v(x,t) = \frac{U(t)L_G}{h \left (g_0L_G - (g_0 - g(t)) x  \right )}  \text{ for } x \in [0, L_G].
\end{equation*}
Thus, the pressure in the glottis $x \in [0, L_G]$ is given by
\begin{equation} \label{IntermediateFormulaForPressure}
\begin{aligned}
   p(x, t) 
   & = p_s(t)  -  \frac{\rho}{2} v(x, t)^2 
   -  \frac{\partial }{ \partial t} U(t)  \left (- \frac{\rho L_G}{h(g_0 - g(t))} 
        \ln{\left (1 - \frac{g_0  - g(t)}{g_0 L_G} x \right )} +  C_{iner}^{(SGT)}\right ) \\
   & = p_s(t)  -  \frac{\rho}{2} v(x,t)^2 - C_{iner}^{(SGT)} U'(t)
   +  \frac{\rho L_G}{h} \cdot  \frac{\partial }{ \partial t}  \left (\frac{U(t)}{g_0 - g(t)} 
        \ln{\left (1 - \frac{g_0 - g(t)}{g_0 L_G} x \right )}   \right ). 
\end{aligned}
\end{equation}
The time derivative term on the right must be manipulated as follows:
\begin{equation*}
\begin{aligned}
  & \frac{\partial }{ \partial t}  \left (\frac{U(t)}{g_0 - g(t)}  \ln{\left (1 - \frac{g_0 - g(t)}{g_0 L_G} x \right )}   \right )
  =  \frac{U'(t)}{g_0 - g(t)}  \ln{\left (1 - \frac{g_0 - g(t)}{g_0 L_G} x \right )} \\
  & + U(t) \left (\frac{g'(t)}{(g_0 - g(t))^2}  \ln{\left (1 - \frac{g_0 - g(t)}{g_0 L_G} x \right )} 
  + \frac{g'(t) x}{g_0 L_G (g_0 - g(t))}  \left (1 - \frac{g_0 - g(t)}{g_0 L_G} x \right )^{-1} \right ) \\
  & = \frac{U'(t)}{g_0 - g(t)}  \ln{\left (1 - \frac{g_0 - g(t)}{g_0 L_G} x \right )} \\
  & + \frac{U(t)g'(t)}{g_0 - g(t)} \left (\frac{1}{g_0 - g(t)}  \ln{\left (1 - \frac{g_0 - g(t)}{g_0 L_G} x \right )} 
  +  \frac{x}{g_0 L_G - (g_0 - g(t))x } \right ) \\
  & =  \frac{U'(t)}{g_0 - g(t)}  \ln{\left (1 - \frac{g_0 - g(t)}{g_0 L_G} x \right )} \\
  & + \frac{U(t)g'(t)}{g_0 - g(t)} \left (\frac{1}{g_0 - g(t)}  \ln{\left (1 - \frac{g_0 - g(t)}{g_0 L_G} x \right )} 
  +  \frac{x}{ g(t) x  + g_0(L_G - x)} \right ).
\end{aligned}
\end{equation*}
Combining this with Eq.~\eqref{IntermediateFormulaForPressure} yields
\begin{equation} \label{FormulaForPressureVolumeVer1}
\begin{aligned}
   p(x, t) 
   & = p_s(t)  -   \frac{\rho L_G^2}{2 h^2\left (g(t) x  + g_0(L_G - x)  \right )^2}  U(t)^2 \\
   & + \left ( \frac{\rho L_G}{h} \frac{1}{g_0 - g(t)}  \ln{\left (1 - \frac{g_0 - g(t)}{g_0 L_G} x \right )}  -  C_{iner}^{(SGT)} \right ) U'(t) \\
   & +  \frac{\rho L_G}{h}  \frac{g'(t)}{(g_0 - g(t))^2} \left (\ln{\left (1 - \frac{g_0 - g(t)}{g_0 L_G} x \right )} 
  +  \frac{1}{1  - \frac{g_0 L_G}{(g_0 -  g(t))x}} \right ) U(t) . 
\end{aligned}
\end{equation}
This equation together with Eq.~\eqref{VolumeVelocityEq} produces the
pressure distribution in glottis, and its first to terms on the RHS of
Eq.~\eqref{FormulaForPressureVolumeVer1} are plainly the steady
incompressible Bernoulli. The third term could be called
\emph{congestion term} (or, water hammer term) as it becomes large and
positive upstream the narrowest part of the glottis just before the
closure when the volume flow is decelerating. The last term vanishes
when $g'(t) = 0$, and for that reason it could be called
\emph{displacement term}.

Since $g(t) \ll g_0$, we may given an approximation
\begin{equation} \label{FormulaForPressureVolumeVer2}
\begin{aligned}
   p(x, t) 
   & = p_s(t)  -   \frac{\rho L_G^2}{2 h^2\left (g(t) x  + g_0(L_G - x)  \right )^2}  U(t)^2 \\
   & - \left ( \frac{\rho L_G}{h g_0}   \ln{\left (\frac{L_G}{L_G - x} \right )}  +  C_{iner}^{(SGT)} \right ) U'(t) \\
   & -  \frac{\rho L_G g'(t)}{hg_0^2} \left (\ln{\left (\frac{L_G}{L_G - x } \right )} 
  +  \frac{x}{L_G - x} \right ) U(t)  
\end{aligned}
\end{equation}
where the congestion and displacement terms have been greatly
simplified.  This equation is best used in conjunction with
Eq.~\eqref{VolumeVelocityApproxEq}. A further simplification can be
carried out by discarding the logarithmic singularity in the
displacement term, or neglecting both the congestion and the
displacement terms entirely.

\ResearchNote{

\begin{remark} At $x = L_G$ we get from Eq.~\eqref{FormulaForPressureVolumeVer1}
\begin{equation*} 
\begin{aligned}
   p(L_G, t) 
   & = p_s(t)  -   \frac{\rho}{2 h^2 g(t)^2}  U(t)^2 \\
   & - \left ( \frac{\rho L_G}{h} \frac{1}{g_0 - g(t)}  \ln{\left (\frac{g_0}{g(t)} \right )}  +  C_{iner}^{(SGT)} \right ) U'(t) \\
   & -  \frac{\rho L_G}{h}  \frac{g'(t)}{(g_0 - g(t))^2} \left (\ln{\left (\frac{g_0}{e g(t)} \right )} 
  +  \frac{g_0}{g(t)} \right ) U(t) . 
\end{aligned}
\end{equation*}
Assuming that $v(t_c) = \lim_{t \to t_c-}{v(t)} = 0$ at the moment of
closure $t_c$, we get
\begin{equation*}
  \lim_{t \to t_c-}{p(L_G,t)} = p_s(t_c) + 
  \frac{\rho L_G}{h g_0}  \lim_{t \to t_c-}{ U'(t) \ln{g(t)}}.
\end{equation*}
Thus, $ \lim_{t \to t_c-}{p(L_G,t)} \leq p_s(t_c)$ and $- \infty <
\lim_{t \to t_c-}{p(L_G,t)}$ implies $\lim_{t \to t_c-}{ U'(t) } = 0$.

Considering Eqs.~\eqref{VolumeVelocityApproxEq} and
\eqref{GlottalIntertanceForWedges}, we see that $v(t_c) = 0$ implies
$-\frac{\rho L_G}{h g_0} \lim_{t \to t_c-} { U'(t)\ln{g(t)}}
 = \lim_{t \to t_c-}{C_{iner}(t) U'(t)} = p_s(t_c) < \infty$.
We conclude that $U'(t) = -\frac{h g_0} {\rho L_G} O \left ( \frac{1}{\ln{g(t)}} \right )$
near $t = t_c$.  It follows that $\lim_{t \to t_c-}{p(L_G,t)} = 0$.

\DisplayNote{JM: check that this conclusion is right. It seems
  incredible somehow.}
\end{remark}

\DisplayNote{JM: You should not use this approximation near $x = L_g$
  where it becomes singular even with an open glottis. For the
  aerodynamic force, you probably could leave out some distance at the
  contact point. Find the gap-dependent interval where the
  approximations in the terms have proportional error less than, say,
  $\epsilon > 0$.}

}
\section{\label{CompressibleSec} Compressible steady flow}

In Sections~\ref{RectangularSec}~and~\ref{WedgeSec}, the convenient
assumption of compressibility was made to derive the equation of
motion for the fluid column in the flow channel. Using the equation of
motion, the expression for the (hydrodynamic) pressure was derived. In
vocal folds models, this pressure produces the aerodynamic forces
leading to the movement of the aerodynamic surfaces.\footnote{ Of
  course, the air flow cannot be fully incompressible since that would
  make the VT and SGT acoustics impossible. The question is whether
  the air flow is incompressible to the extent that it is a
  \emph{reasonable approximation} for treating the total inertia and
  the resulting aerodynamic force to flow channel walls.}

To have an educated opinion on this matter, one must explicitly deal
with some form of compressible flow not accounted by the acoustic
approximation. The challenge here is that the treatment of a general
\emph{nonsteady} (in)compressible flow is not possible using
elementary mathematical tools and solutions in a closed form. Thus, we
consider only the steady variant of a flow of an ideal gas column. In
this case, the modifications due to isentropic thermodynamics are
well-known. We give two examples on a steady flow having the physical
dimensions resembling the glottal flow.

\subsection{Generalities on isentropic ideal gas flow}

We assume that the usual isentropic relations hold
\begin{equation} \label{IsentropicRelations}
  \frac{p}{p_0} = \left (\frac{T}{T_0} \right )^{\frac{\gamma}{\gamma - 1}} \quad \text{ and } \quad
  \frac{\rho}{\rho_0} = \left (\frac{T}{T_0} \right )^{\frac{1}{\gamma - 1}}.
\end{equation}
Thus 
\begin{equation*}
  \frac{p}{\rho} =  \left (\frac{T}{T_0} \right )^{\frac{\gamma}{\gamma - 1}} \left (\frac{T}{T_0} \right )^{\frac{-1}{\gamma - 1}}   \frac{p_0}{\rho_0} 
= \frac{T}{T_0} \cdot   \frac{p_0}{\rho_0} 
\end{equation*}
which makes the temperature distribution easier to compute than
pressure or density distributions. Actually, it is just the usual equation of state for ideal gas. For the speed of sound, we get
\begin{equation*}
  c^2 = \gamma \frac{p}{\rho} = \frac{T}{T_0} \cdot \gamma
  \frac{p_0}{\rho_0}.
\end{equation*}

A compressible, steady Bernoulli flow inside insulated streamlines is
described by
\begin{equation} \label{CompressibleBernoulli}
  \frac{1}{2}v^2 + \frac{\gamma}{\gamma -1} \frac{p}{\rho} = \frac{\gamma}{\gamma -1} \frac{p_0}{\rho_0}.
\end{equation}
The conservation of mass in a tube (i.e., nozzle) whose intersectional
area is $A = A(x)$, given by
\begin{equation*}
  \rho v A = V_m
\end{equation*}
where $V_m$ is the mass flow, considered to be constant of time. We,
of course, make the assumption that $A(x)$ is ``slowly varying'' in
the sense that an isentropic, compressible flow can be supported in
the entire inner volume of the tube (i.e., tube walls are always
streamlines at least in the subsonic part of the nozzle).

Putting these together
\begin{equation*}
  \frac{1}{2} \left ( \frac{V_m}{A} \right )^2 + \frac{\gamma}{\gamma
    -1} p \rho - \frac{\gamma}{\gamma -1} \frac{p_0}{\rho_0} \rho^2 =
  0.
\end{equation*}
Now, from the isentropic relations we get
\begin{equation*}
  p \rho = p_0 \rho_0 \left (\frac{T}{T_0} \right )^{\frac{\gamma + 1}{\gamma - 1}} \text{ and }
  \rho^2 = \rho_0^2 \left (\frac{T}{T_0} \right )^{\frac{2}{\gamma - 1}}.
\end{equation*}
Plugging in, we get
\begin{equation} \label{TemperatureDistributionEq}
    \frac{1}{2} \left ( \frac{V_m}{A(x)} \right )^2 +
    \frac{\gamma p_0 \rho_0 }{\gamma -1} \left [ \left (\frac{T(x)}{T_0} \right
    )^{\frac{\gamma + 1}{\gamma - 1}} -  \left (\frac{T(x)}{T_0} \right
    )^{\frac{2}{\gamma - 1}} \right ] = 0
\end{equation}
over the length of the tube. Note that $\frac{\gamma + 1}{\gamma - 1}
= 1+ \frac{2}{\gamma - 1}$.

Another form for \eqref{TemperatureDistributionEq}
is given by
\begin{equation*}
\begin{aligned}
  v(x)^2 & =  \left ( \frac{V_m}{A(x) \rho(x)} \right )^2 \\
  & = \frac{
    \rho_0^2}{\rho(x)^2} \cdot \frac{2}{\gamma -1}
  \cdot \gamma \frac{p_0}{\rho_0}\frac{T(x)}{T_0} \cdot \left [ \left (\frac{T(x)}{T_0} \right
    )^{\frac{2}{\gamma - 1}-1} - \left (\frac{T(x)}{T_0} \right
    )^{\frac{\gamma + 1}{\gamma - 1}-1}\right ] \\
  & = 
  \left (\frac{T(x)}{T_0} \right )^{-\frac{2}{\gamma - 1}}\cdot 
  \frac{2}{\gamma -1} \cdot c(x)^2 \cdot   
  \left [ \left (\frac{T(x)}{T_0} \right
    )^{\frac{2}{\gamma - 1}-1} - \left (\frac{T(x)}{T_0} \right
    )^{\frac{2}{\gamma - 1}}\right ] \\
& = 
 \frac{2}{\gamma -1} \cdot c(x)^2 \cdot   
\left [ \left (\frac{T(x)}{T_0} \right
   )^{-1} - 1 \right ] 
\end{aligned}
\end{equation*}
which leads to the expression for the Mach number
\begin{equation*}
  M(x)^2 = \frac{v(x)^2}{c(x)^2} =  \frac{2}{\gamma -1} \left [\frac{T_0}{T(x)} - 1 \right ].
\end{equation*}
The speed of sound is reached at the temperature $T(x) = \frac{2}{\gamma + 1} T_0$ when 
\begin{equation*}
  p(x) = \left (\frac{2}{\gamma + 1} \right )^{\frac{\gamma}{\gamma - 1}} p_0 \text{ and }
\rho(x) = \left (\frac{2}{\gamma + 1} \right )^{\frac{1}{\gamma - 1}} \rho_0.
\end{equation*}
The condition on the nozzle area function for reaching Mach 1 at $x$
is given by
\begin{equation} \label{SupersonicNozzleArea}
\begin{aligned}
  A(x)^2 & =\frac{V_m^2}{\gamma p_0\rho_0} \frac{\rho_0^2}{ \rho(x)^2} \frac{T_0}{T(x)} \\
  & = \frac{V_m^2}{\gamma p_0\rho_0} \left
  (\frac{T_0}{T(x)} \right )^{\frac{2}{\gamma - 1}} \frac{T_0}{T(x)} =
  \frac{V_m^2}{p_0\rho_0}\cdot \frac{1}{\gamma} \left
  (\frac{T_0}{T(x)} \right )^{\frac{\gamma + 1}{\gamma - 1}} \\ & =
  \frac{1}{\gamma}\left (\frac{\gamma + 1}{2} \right )^{\frac{\gamma +
      1}{\gamma - 1}} \cdot \frac{V_m^2}{p_0\rho_0} =
\frac{\gamma + 1}{\gamma}\left (1 + \frac{\gamma - 1}{2} \right )^{\frac{2}{\gamma - 1}} \cdot \frac{V_m^2}{2 p_0\rho_0}
\end{aligned}
\end{equation}
which yields the \emph{critical area} $A_{sonic}$.  The speed of sound
at Mach 1 is given by
\begin{equation*}
  c(x)^2 = \gamma \frac{p_0}{\rho_0}\frac{T(x)}{T_0} = c_0^2 \cdot  \frac{2}{\gamma + 1} \text{ with } c_0^2 = \gamma \frac{p_0}{\rho_0}.
\end{equation*}

\begin{remark}
  Observe that the incompressible limit case is obtained in
  Eqs.~\eqref{IsentropicRelations}--\eqref{CompressibleBernoulli} by
  letting $\gamma \to \infty$. In particular,
  Eq.~\eqref{IsentropicRelations} gives in the limit
  \begin{equation*}
    \frac{p}{p_0} = \frac{T}{T_0}  \quad \text{ and } \quad
  \frac{\rho}{\rho_0} = 1.
  \end{equation*}
  Of course, the speed of sound $c \to \infty$ as $\gamma \to \infty$ as
  well, and hence the Mach number $M \to 0$ for any fixed finite mass
  flow $V_m$. 

  However, observe that
  \begin{equation*}
    A_{sonic} \to  \frac{V_m}{\sqrt{2 p_0\rho_0}} \quad \text{ as } \gamma \to \infty
  \end{equation*}
  which is a rather peculiar observation since one would expect to
  have $A_{sonic} = 0$ for an incompressible fluid. The incompressible
  steady Bernoulli gives $\frac{1}{2}v^2 + \frac{p}{\rho} =
  \frac{p_0}{\rho_0}$ which gives an upper limit for the velocity $v$
  since $p/\rho \geq 0$. This limit corresponds to the limit of the
  sonic area where the Venturi effect has reached vacuum.
\end{remark}

\subsection{Two examples on diatomic ideal gas}

For diatomic ideal gas, $\gamma = 7/5$, and hence $\gamma/(\gamma - 1)
= 7/2$, $2/(\gamma -1) = 5$, $2/(\gamma + 1) = 5/6$, and $(\gamma +
1)/(\gamma -1) = 6$.  With these values,
Eq.~\eqref{TemperatureDistributionEq} becomes
\begin{equation*} 
    \left ( \frac{V_m}{A(x)} \right )^2 +
    7 p_0 \rho_0  \left [ \left (\frac{T(x)}{T_0} \right
    )^{6} -  \left (\frac{T(x)}{T_0} \right
    )^{5} \right ] = 0,
\end{equation*}
or, equivalently,
\begin{equation*} 
    \frac{V_m^2}{7 p_0 \rho_0} =  A(x)^2 \left (\frac{T(x)}{T_0} \right
    )^{5} \left [1 - \frac{T(x)}{T_0}  \right ].
\end{equation*}
From this it follows that $T(x), T_0 > 0$ implies $T(x) \leq T_0$.
The restriction of remaining subsonic takes the form $T(x) >
\frac{2}{\gamma + 1} T_0 = \frac{5}{6}T_0$. We proceed to maximize the
function $b \mapsto b^5(1-b)$ for $b \in [5/6, 1]$ corresponding the
subsonic regime. Differentiating and setting $5b^4 - 6b^5 = b^4(5 -
6b) = 0$ leads to $b = 5/6$. So, the minimum area $A(x)$ consistent
with the compressible Bernoulli principle and the isentropic process
takes place when Mach number $M(x) = 1$ is reached.  It also follows
that $T(x)$ is a decreasing (increasing) function of $A(x)$ in the
subsonic (supersonic) regime.

For diatomic ideal gas, the critical ``sonic area'' is
\begin{equation*}
  A_{sonic} \approx 1.4604 \cdot \frac{V_m}{\sqrt{p_0 \rho_0}}.
\end{equation*}
At Mach 1, the speed of sound has been reduced by the factor of
$\sqrt{5/6} \approx 0.913$. Let us now make two computations to
estimate what kind of intersection areas lead to Mach 1 and Mach 0.3
flows for parameter values typical of the glottal flow.

\begin{example}  \label{FirstCompressibleExample}
  Consider a flow of $2 \, \mathrm{dl}/\mathrm{s}$ of air with $T_0 =
  300 \, \mathrm{K}$, $p_0 = 100 \, \mathrm{kPa}$, and $\rho_0 = 1.2
  \, \mathrm{kg}/\mathrm{m}^3$. In these conditions, the speed of
  sound is $342 \, \mathrm{m}/\mathrm{s}$.

  Then the mass flow is $V_m = 2.4\cdot 10^{-4} \, \mathrm{kg}/s$, and
  $A_{sonic} = 1.01\cdot 10^{-6} \, \mathrm{m}^2 = 1.01 \,
  \mathrm{mm}^2$.  The speed of sound at the narrow point is $c =
  \sqrt{\frac{5}{6}}\cdot 342 \, \mathrm{m}/\mathrm{s} = 312 \,
  \mathrm{m}/\mathrm{s}$, temperature $\frac{5}{6} \cdot 300 \mathrm{K} =  250 \, \mathrm{K}$, and
  pressure $52.8 \, \mathrm{kPa}$.

  If the same volume flow was incompressible through the same area,
  the speed would be $v = 2 \, \mathrm{dl}/\mathrm{s} / 1.01 \,
  \mathrm{mm}^2 = 198 \, \mathrm{m}/\mathrm{s}$.
\end{example}

\begin{example} \label{SecondCompressibleExample}
  Consider a flow of $2 \, \mathrm{dl}/\mathrm{s}$ of air with $T_0 =
  300 \, \mathrm{K}$, $p_0 = 100 \, \mathrm{kPa}$, and $\rho_0 = 1.2
  \, \mathrm{kg}/\mathrm{m}^3$. 

  Again, the mass flow is $V_m = 2.4\cdot 10^{-4} \,
  \mathrm{kg}/s$. At  Mach $0.3$, the temperature of the gas would be
  \begin{equation*}
    T(x) = \frac{T_0}{\frac{\gamma -1}{2}M(x)^2 + 1} = \frac{300 \,
      \mathrm{K}}{\frac{1}{5} 0.3^2 + 1} = 295 \, \mathrm{K}.
  \end{equation*}
  The speed of sound at this temperature is $c(x) = \sqrt{295/300}
  \cdot 342 \, \mathrm{m}/\mathrm{s} = 339 \,
  \mathrm{m}/\mathrm{s}$. Since we are going Mach 0.3, the speed of
  the air is $102 \, \mathrm{m}/\mathrm{s}$. The density is given by
  $\rho(x) = \rho_0 \left (\frac{T(x)}{T_0} \right )^{\frac{1}{\gamma
      - 1}} = 1.2 \, \mathrm{kg}/\mathrm{m}^3 \cdot \left (\frac{295
    \, \mathrm{K}}{300 \, \mathrm{K}} \right )^{\frac{5}{2}} =  1.15 \, \mathrm{kg}/\mathrm{m}^3$. Now,
  \begin{equation*}
    A(x)  = \frac{V_m}{\rho(x)v(x)} = \frac{2.4\cdot 10^{-4} \,
  \mathrm{kg}/s}{1.15 \, \mathrm{kg}/\mathrm{m}^3 \cdot 102 \, \mathrm{m}/\mathrm{s}}
    = 2.05 \, \mathrm{mm}^2.
  \end{equation*}

  If the same volume flow was incompressible through the same area,
  the speed would be $v = 2 \, \mathrm{dl}/\mathrm{s} / 2.05 \,
  \mathrm{mm}^2 = 97.6 \, \mathrm{m}/\mathrm{s}$.
\end{example}

Considering
Examples~\ref{FirstCompressibleExample}~and~\ref{SecondCompressibleExample},
conclude that for glottal opening areas over $2 \, \mathrm{mm}^2$ the
usual ``rule of thumb'' value Mach 0.3 holds and the flow can be
regarded as incompressible (even neglecting all viscosity and
nonsteadyness effects). For human glottis, the area of $2 \,
\mathrm{mm}^2$ can be considered quite a small opening. Just by
halving the opening area to $1 \, \mathrm{mm}^2$, the hypothetical
compressible flow already gets supersonic which certainly is
counterfactual as far as the glottal flow is concerned. (The adiabatic
cooling to $250$ K would perhaps freeze the vocal folds, and a
pressure drop to half of the ambient would be destructive as well.)
Realistic pressure loss effects (such as the Poiseuille law for
viscous laminar flow) will check the flow velocity at the narrowest
position much before the flow gets supersonic.

\subsection{Discussion on energy losses in compressible flows}

The internal friction due to fluid viscosity leads to a pressure loss
that can be modelled by Poiseuille's law for flow channels having
circular intersections.  Other intersectional geometries, such as
rectangular and triangular shapes, can be given analytic descriptions,
and also they are known as Poiseuille's law.  For even more general
geometries, one is compelled to use heuristic approximations of, e.g.,
hydrodynamic radii or numerical solutions. These variant of
Poiseuille's law are derived for an incompressible, laminar, steady
flow, too.

When the flow is compressible yet remains isentropic, additional
mechanisms for kinetic energy and pressure loss emerge. There is
temperature variation (adiabatic heating) at stagnation points. If
such heat is conducted to surrounding structures in a lower
temperature (in which case the thermodynamic system is not perfectly
adiabatic), there is a total energy loss from the fluid at that
position. In terms of the compressible, steady Bernoulli flow, the
energy loss shows up as an \emph{unrecoverable} pressure loss.  There
are reasons to believe that Poiseuille's law alone is not a sufficient
description of unrecoverable pressure loss in glottal flow.






A nonexhaustive list of mechanisms that could result in a pressure
loss in addition to the viscosity effects:
\begin{enumerate}
\item We could have \emph{adiabatic heating at a stagnation point} or
  \emph{adiabatic cooling in a constriction}. In itself, these effects
  do not (by definition) amount to loss of heat from the fluid in a
  perfectly insulated flow channel.

  In the first case, the area function changes shape in such a way
  that a streamline actually ends inside the tube. Such a phenomenon
  would surely take place in a rectangular glottis when the flow meets
  the glottal wall perpendicularly. Then, the temperature would
  increase to $T_0$ (but not higher!) at the stagnation point. If the
  channel walls at that point are at a lower temperature (driven
  there, e.g., by the adiabatic cooling due to the flow effect as
  described above), the heat conduction would actually lead to a loss
  of energy from the flow. Conversely, the temperature at a
  constriction may get lower than the wall temperature. Then the heat
  conduction from the wall would increase the total energy of the
  fluid. Note that the heat capacity of the tissue walls would be much
  higher than that of air.

  Not that the adiabatic cooling at a constriction is offset by the
  opposite effect of heat production due to viscosity. It is, of
  course, strictly speaking wrong to consider \emph{adiabatic} effects
  in such a viscous flow.

  Based on the figures given in
  Example~\ref{SecondCompressibleExample}, the temperature variation
  due to adiabatic heating and/or cooling is about $1$ \% of the
  absolute temperature. The effects to the pressure are of the same
  order due to the thermodynamic equation of state.  The heat exchange
  with the walls seems a very complicated phenomenon, and it is
  perhaps not an effect that needs accounting for in low-order models.
\item There could be a \emph{dissipative boundary layer effect}
  transforming kinetic energy into heat by the viscosity of the fluid,
  not accounted by the Poiseuille's law. 
\item There could be some form of \emph{vortex formation} that would
  ultimately transform kinetic energy into heat via internal viscous
  losses in the fluid.  This would not be a boundary layer effect, nor
  accounted for by the Poiseuille's law for laminar flow.
\item \emph{Dynamic effects} that would depend on the moving walls,
  nor observed in model experiments with rigid walls.
\end{enumerate}

It seems likely that properly tuned variants of Poisseuille's law
serve well as a first approximation for modelling of the glottal
pressure loss. At least, it has a strong theoretical background in the
laminar flow regime if the compressibility is not an issue either. The
``correction terms'' can be motivated by experimental work on physical
models (such as the work by van den Berg \& al., Fulcher \& al., etc.)
or deep numerical computations involving Navier--Stokes -based flow
models with thermodynamical coupling such as
\cite{S-M-M:CFSPURVTG}. However, the author has no knowledge of such
computational work. 


\section{\label{WaveGuideSec} Inertia and termination of a waveguide }

In Sections~\ref{RectangularSec}~and~\ref{WedgeSec}, we derived
equations of motion for an \emph{incompressible fluid column} in a
tubular domain where part of the domain boundary was allowed to change
as a function of time.  We observed that both of these (lossless)
models lead to the equation of motion \eqref{EquationOfMotion} where
the only difference is the expression of the glottal inertance term
$C_{iner}^{(G)}(t)$; see
Eqs.~\eqref{RectangularInertance}~and~\eqref{GlottalIntertanceForWedges}.
In both cases, the vocal tract inertance is given by the integral
\begin{equation*}
  C_{iner}^{(VT)} := \rho \int_{0}^{L_{VT}} {\frac{ds}{A_{VT}(s)}}.
\end{equation*}
Since acoustics is also about the movement of gas molecules (albeit by
rather small distances around an equilibrium position), it is natural
to ask if the same expression $C_{iner}^{(VT)}$ plays a role in
acoustic equations in the same vocal tract volume. The purpose of this
section is to answer this question.

For simplicity, let us consider a fluid column of length $\ell > 0$
with uniform intersectional area $A$. The fluid density and the speed
of sound are denoted by $\sigma$ and $c$. The assumption that $c <
\infty$ means, of course, that the fluid is to some degree
compressible.

\subsection{\label{WaveGuideSubSec} Generalities}

Consider a waveguide of constant intersectional area $A$ and
characteristic impedance $Z_0 = \rho c / A$.  Its \emph{acoustic
impedance} is given by
\begin{equation} \label{WaveGuideImpedance}
  Z_{ac}(s) = Z_0 \frac{Z_L(s) + Z_0 \tanh{s T_0}}{ Z_0 + Z_L(s) \tanh{s T_0}}
\end{equation}
where we denote the \emph{transmission time} by $T_0 = \ell / c$ and
the \emph{termination impedance} by $Z_L(s)$. We assume that both
$Z_L(s)$ and the admittance $Z_L(s)^{-1}$ are analytic, positive real
functions on $\C_+$. Another form for the impedance is
\begin{equation*}
  Z_{ac}(s) = Z_0 \frac{Z_L(s) \cosh{s T_0} + Z_0 \sinh{s T_0}}{ Z_0
    \cosh{s T_0} + Z_L(s) \sinh{s T_0}}
\end{equation*}
which has the merit of having the numerator and the denominator
analytic in $\C_+$. There is an immediate conclusion:
\begin{proposition}
  If the termination load has $Z_L(s)$ has a zero at $s = \frac{k
    \pi}{T_0}$, $k \in Z$, then so does have the acoustic impedance
  $Z_{ac}(s)$.
\end{proposition}
\noindent Indeed, $0 = 2\sinh{s T_0} = e^{s T_0} - e^{-s T_0}$ implies
$e^{2s T_0} = 1$ which is possible if and only if $s = \frac{k
  \pi}{T_0}$. This observation appears to be of particular value when
$k = 0$ since then it implies the additivity of inertances\footnote{I
  think any Webster's resonator has a discrete number of frequencies
  where the zeroes of the load show through.}.

Let us begin by studing the purely resistive termination for observing
the qualitative behaviour of waveguide resonances and
antiresonances. If $Z_L(s) = R_L$, the numerator and the denominator
are both entire functions. In this case, the zeroes of $Z_{ac}(s)$ are
given by
\begin{equation*}
  -\frac{R_L}{Z_0} = \tanh{s T_0} = \frac{e^{s T_0} - e^{-s T_0}}{e^{s T_0} + e^{-s T_0}} = \frac{e^{2s T_0} - 1}{e^{2s T_0} + 1}; 
\end{equation*}
that is, 
\begin{equation*}
  e^{-2xT_0} \left (\cos{2yT_0} - i \sin{2yT_0}  \right ) = e^{-2sT_0} = \frac{Z_0 + R_L}{Z_0 - R_L} 
\end{equation*}
where $s = x + y i$. We conclude that $\sin{2yT_0} = 0$, i.e., $2y T_0
= n \pi$ for $n \in \Z$. Then $\cos{2yT_0} = (-1)^n$, and the equation
becomes
\begin{equation*}
  (-1)^n \frac{Z_0 + R_L}{Z_0 - R_L} = e^{-2xT_0}.
\end{equation*}
If $Z_0 > R_L$, i.e., $n = 2k$, we get for the zeroes
\begin{equation*}
  s= -\frac{1}{2 T_0} \ln{\abs{\frac{Z_0 + R_L}{Z_0 - R_L}}} + i k \frac{\pi}{T_0}.
\end{equation*}
If $Z_0 < R_L$, i.e., $n = 2k + 1$, we get
\begin{equation*}
  s= -\frac{1}{2 T_0} \ln{\abs{\frac{Z_0 + R_L}{Z_0 - R_L}}} + i \left
  (k + \frac{1}{2} \right ) \frac{\pi}{T_0}.
\end{equation*}
A similar computation with a similar outcome can be carried out for
the poles of $Z_{ac}(s)$.  If $Z_0 = R_L$, there are no zeroes nor
poles as expected.

\begin{remark} \label{ResonanceResistanceRemark}
We observe that introducing resistance to termination does not change
the resonant frequencies nor antinodal frequencies (as far as we have
$Z_0 \neq R_L$) but it will add losses to the impedance/admittance
system. For acoustic waveguides with nonuniform intersectional areas,
the resonant frequencies do depend on the termination resistances.

In all cases, the resonant frequencies and the antinodal frequencies
are sensitive to inductive or capacitive loading at the termination.
\end{remark}

There is one more detail whose statement has some system theoretical
interest.  Clearly, the impedance $Z_{ac}(s)$ is a positive-real
transfer function of an impedance passive system for any positive-real
termination impedance $Z_L(s)$. Is the system well-posed in the usual
infinite-dimensional linear systems sense? For general impedance
passive systems, the well-posedness is equivalent with the fact that
the transfer function is uniformly bounded on some vertical line $x_0
+ i \R = \{ x_0 + y i \, : \, y \in \R \}$ for $x_0 > 0$ (Theorem 5.1
in Staffans (2002)). Let us proceed to check this condition for a
transmission line with a constant area function.\footnote{I think also
  the general impedance conservative waveguide with non-constant area
  function -- described by Webster's partial differential equation --
  is well-posed but proving it would require a completely different
  and much more difficult approach.}

\begin{proposition} \label{TransmissionLineWPProp}
  For any rational, positive-real analytic function $Z_L(s)$, the
  transfer function $Z_{ac}(s)$ given by
  Eq.~\eqref{WaveGuideImpedance} is bounded on a vertical line lying
  in the open right half plane $\C_+$.
\end{proposition}
\begin{proof}
It is clearly enough to show that the function
\begin{equation} \label{AuxiliaryFunction}
    G(s) = \frac{Z(s) + \tanh{s}}{1 + Z(s) \tanh{s}}
\end{equation}
is bounded on such a vertical line where $Z(s)$ is a positive-real
analytic function in $\C_+$. We first need an observation concerning
the hyperbolic tangent, namely that
\begin{equation*}
  \mathop{Re} \tanh{(x + y i)} = \frac{\sinh{2x}}{\cosh{2x} + \cos{2y}}, \quad x,y \in \R, 
\end{equation*}
implying the estimate
\begin{equation} \label{HyperTanEstimate}
0 < 1 - \frac{1}{\sinh{2x} + 1} <  \mathop{Re} \tanh{(x + y i)} < 1 + \frac{1}{\cosh{2x} - 1}
\end{equation}
for all $x > 0$ since $\sinh{2x} < \cosh{2x}$ and $\abs{\cos{2y}} \leq
1$. In particular, $\mathop{Re} \tanh{s} > 0$ for $s \in \C_+$. A
similar computation shows that
\begin{equation*}
  \mathop{Re} \frac{1}{\tanh{(x + y i)}} = \frac{\sinh{2x}}{\cosh{2x} - \cos{2y}}, \quad x,y \in \R, 
\end{equation*}
leading to exactly the same upper and lower bounds for $\mathop{Re}
\frac{1}{\tanh{(x + y i)}}$ as are given for $\mathop{Re} \tanh{(x + y
  i)}$ in Eq.~\eqref{HyperTanEstimate}. We conclude that both
$\mathop{Re} \tanh{s}$ and $\mathop{Re} \frac{1}{\tanh{s}}$ are uniformly bounded from above and below on all vertical
lines $x_0 + i \R  \subset \C_+$ by a strictly nonnegative
constant.

\begin{enumerate}

\item{{\bf Case where $1/Z(s)$ is bounded on $x_0 + i \R$ for some $x_0 > 0$.}}

Because 
\begin{equation*}
  G(s) = \frac{1}{Z(s)} \left (1 + \left (1 + \frac{1}{Z(s)^2} \right ) \frac{1}{\frac{1}{Z(s)} + \tanh{s}} \right ),
\end{equation*}
it is enough to show that $(\frac{1}{Z(s)} + \tanh{s})^{-1}$ is
uniformly bounded from above on $x_0 + i \R$. Now
\begin{equation*}
  \abs{\frac{1}{Z(s)} + \tanh{s}} = \abs{\left (-\frac{1}{Z(s)}\right ) - \tanh{s}} >  \mathop{Re} \tanh{s} > 0 \quad \text{ for  } \quad s \in \C_+
\end{equation*}
since $-1/Z(s) \in \C_-$ and $\tanh{s} \in \C_+$. We conclude from
Eq.~\eqref{HyperTanEstimate} that
\begin{equation*}
  \frac{1}{\abs{\frac{1}{Z(s)} + \tanh{s}}} < \frac{1}{\mathop{Re} \tanh{s}}
< \frac{1}{1 - \frac{1}{\sinh{2x} + 1}} = 1 + \frac{1}{\sinh{2x}}
\end{equation*}
where $s = x + yi$, $x > 0$ and $y \in \R$ arbitrary.

\item{{\bf Case where $Z(s)$ is bounded on $x_0 + i \R$ for some $x_0 > 0$.}}

We now write 
\begin{equation*}
  G(s) =   \left ( \frac{Z(s)}{\tanh{s}} + 1 \right ) \left (\frac{1}{\tanh{s}} + Z(s) \right )^{-1}
\end{equation*}
and observe that $\frac{Z(s)}{\tanh{s}}$ is uniformly bounded from
above on $x_0 + i \R$ by using the lower estimate in
Eq.~\eqref{HyperTanEstimate}. We proceed to show that $\left (
\frac{1}{ \tanh{s}} + Z(s) \right )^{-1}$ is uniformly bounded on $x_0
+ i \R$. This time
\begin{equation*}
  \abs{\frac{1}{\tanh{s}} + Z(s)} = \abs{\left
    (-\frac{1}{\tanh{s}}\right ) - Z(s)} > \mathop{Re}
  \frac{1}{\tanh{s}} \quad \text{ for  } \quad s \in \C_+
\end{equation*}
 since $-1/\tanh{s} \in \C_-$ and $Z(s) \in
\C_+$. Again, we obtain the estimate
\begin{equation*}
\frac{1}{\abs{\frac{1}{\tanh{s}} + Z(s)}} < 1 + \frac{1}{\sinh{2x}}
\end{equation*}
where $s \in \C_+$ and $x = \mathop{Re} s$.
\end{enumerate}

We proceed to the case where $Z(s)$ is a rational function. Then
either $Z(s)$ or $1/Z(s)$ is proper. A proper transfer function is
bounded on some right half plane $x_0 + \C_+$, $x_0 > 0$. The claim of
the proposition now follows from the previous two special cases.
\end{proof}

\subsection{\label{InertialLimitSubSec} Inertial limit}

 Assume that there is a piston at the input end of the tube moving at
 the velocity $v(t) = v_0 \sin{k t}$ and acceleration $a(t) = v'(t)$.
 Then the \emph{inertial (counter) pressure} for $0 < k \ll 1$ is
 plainly given by the Newton's second law for an (nearly)
 \emph{incompressible} fluid is
\begin{equation*}
  p(t) = \frac{F(t)}{A} = \frac{m v'(t)}{A} = \frac{\rho \ell A \cdot v_0 k \cos{kt}}{A}
=  \rho \ell \cdot v_0 k \cos{kt}
\end{equation*}

from which for the inertial impedance transfer function (from volume
velocity to pressure) we get
\begin{equation} \label{InertialImpedance}
  Z_{iner}(s) = \frac{\rho \ell k v_0 \left ( \frac{s}{s^2 + k^2} \right ) }{ A v_0\left (\frac{k}{s^2 + k^2} \right )}
  = \frac{\rho \ell}{A} s = C_{iner} s
\end{equation}
where $C_{iner} = \rho \ell / A$ is the \emph{inertance} of the fluid
column having a constant intersection area $A$.

For a \emph{compressible} fluid in a column of same dimensions,
terminated to an acoustic impedance $Z_L(s)$, we get
\begin{equation*}
  Z_{ac}(s) = Z_0 \frac{Z_L(s) + Z_0 \tanh{s T_0}}{ Z_0 + Z_L(s) \tanh{s T_0}}
\end{equation*}
where the characteristic impedance is given by $Z_0 = \frac{\rho
  c}{A}$ and the transmission time $T_0 = \frac{\ell}{c}$. Note that
\begin{equation*}
  \lim_{s \to 0}{Z_{ac}(s)} = Z_L(0)
\end{equation*}
since $\tanh{0} = 0$.

\begin{definition}
We say that the termination impedance transfer function $Z_L(s)$ is
\emph{inertially feasible} if 
\begin{equation*}
  \lim_{s \to 0}{\frac{Z_{ac}(s)}{Z_{iner}(s)}} = r_{iner} \in \R.
\end{equation*}
In this case, the number $r_{iner}$ is called \emph{inertial factor}.
\end{definition}
Obviously, a necessary condition for inertial feasibility is that
$Z_L(0) = 0$. For perfectly terminated waveguides $Z_L(s) = Z_0$ we
have $Z_{ac}(s) = Z_0$ which is not an inertially feasible
termination.

Let us then compute the value of $r_{iner}$. We have
\begin{equation*}
    \frac{Z_{ac}(s)}{Z_{iner}(s)} 
    = \frac{A}{\rho \ell} \cdot \frac{Z_L(s) + Z_0 \tanh{s T_0}}{s} \cdot \frac{Z_0}{Z_0 + Z_L(s) \tanh{s T_0}}
\end{equation*}
The last term approaches to $1$ as $s \to 0$. By l'Hospital's rule, we get
\begin{equation*}
  \lim_{s \to 0}{\frac{Z_L(s) + Z_0 \tanh{s T_0}}{s}} 
  =  \lim_{s \to 0}{\left ( Z_L'(s) + Z_0 T_0 (1 - \tanh^2{s T_0}) \right )} =
  Z_L'(0) + Z_0 T_0.
\end{equation*}
Thus, 
\begin{equation}
  r_{iner} = \frac{A Z_L'(0)}{\rho \ell} + \frac{A}{\rho \ell}
  \frac{\rho c}{A} \frac{\ell}{c} = 1 + \frac{A Z_L'(0)}{\rho \ell} =
  1 + \frac{Z_L'(0)}{Z_0 T_0}.
\end{equation}
Using the inertance, we get yet another formula
\begin{equation} \label{InertanceExpression}
  r_{iner} = 1 + \frac{Z_L'(0)}{C_{iner}}. 
\end{equation}
Clearly, $r_{iner} \geq 1$ and $r_{iner} = 1$ if and only if $Z_L(0) =
Z_L'(0) = 0$.

Given the termination impedance $Z_L(s)$, the inertial factor tells
the proportion how much a fluid column of length $\ell$ must be
extended in order it to have the same inertia as a comparable
transmission line of length $\ell$ when terminated to $Z_L(s)$.

\begin{example} \label{RationalBoundaryImpedanceExample}
Let us consider a commonly used boundary impedance model, namely a
resistance and an inductance in parallel.  Then $R_L(s) = \frac{s
  RL}{R + s L}$ and $R_L'(0) = L$. We get
\begin{equation*}
  r_{iner} = 1 + \frac{AL}{\rho \ell}  = 1 + \frac{L c}{Z_0\ell} = 1 + \frac{L}{ Z_0 T_0}.
\end{equation*}
\end{example}

\begin{remark}
By Eq.~\eqref{InertanceExpression}, the acoustic inertance of a
waveguide can be tuned by external inductive loading to any value
larger than $C_{iner} = \rho \ell /A$. In fact, the rational impedance
of Example~\ref{RationalBoundaryImpedanceExample} is a sufficiently
rich class of acoustic terminations for this purpose.

However, the inductive loading not only increases the acoustic
inertance. It also moves the positions of resonant frequencies of the
system. This is inconvenient when the terminated waveguide acts as an
acoustic load in a larger system for which both the inertance and the
resonant frequencies must be controlled to some predetermined target
values.
\end{remark}

\begin{remark}
It was already pointed out that the termination of an uniform diameter
acoustic waveguide to its characteristic impedance will not result in
an inertially feasible acoustic load. This is understandable since
such a waveguide appears to be infinity long with infinite mass, and
its inertance cannot be expected to have any finite value.  In
general, acoustically nonreflecting boundary termination of \emph{any}
does not seem to be inertially feasible. However, an absorbing
boundary condition may well be a desirable feature in an acoustic
(part of a) model.
\end{remark}

\subsection{Inertial proportion}

The proportion of the characteristic impedance and the inertive
response of a wave guide is called \emph{inertial proportion} at
frequency $f$, given by
\begin{equation*}
  P_{iner}(f) = \frac{C_{iner}}{Z_0}\cdot 2 \pi f = (\rho \ell / A)
  \cdot (A / \rho c) \cdot 2 \pi f = \frac{2 \pi \ell f}{c} = \frac{2
    \pi \ell}{\lambda}
\end{equation*}
where $\lambda = c / f$ is the wavelength. For low frequencies $f$,
the number $2 \pi C_{iner} f$ approximates the acoustic impedance
$\abs{Z_{ac}(2 \pi f i)} = Z_0 \abs{\tanh{2 \pi T_0 f i}} = Z_0
\abs{\tan{2 \pi \ell/\lambda}} $ of the transmission line of length
$\ell$, terminated to a short circuit. Thus
\begin{equation*}
  P_{iner}(f) \approx \tan{\frac{2 \pi \ell}{\lambda}} = \tan{P_{iner}(f)}.
\end{equation*}
In the case of the VT, it is typical to use the value $\ell = 0.17 \,
\mathrm{m}$ and $c = 340 \, \mathrm{m}/\mathrm{s}$.  The numerical
values of $P_{iner}(f)$ and $\tan{\frac{2 \pi \ell}{\lambda}}$ for
speech relevant frequencies $f$ are given in Table~\ref{InertialPropTable}.

\begin{table}[h]
 \centerline{
    \begin{tabular}{|l|c|c|c|c|c|c|c|c|c|c|c|}
\hline
$f$ in \textrm{Hz}                 & $50$ & $75$ & $100$ & $125$ & $150$ & $175$ & $200$ & $225$ & $250$    \\
\hline
$P_{iner}$                         & 0.1571 &   0.2356 &   0.3142 &   0.3927 &   0.4712  &  0.5498  &  0.6283 &   0.7069 &   0.7854 \\
$\tan{P_{iner}(f)}$                & 0.1584  &  0.2401  &  0.3249  &   0.4142  &  0.5095  &  0.6128  &  0.7265  &  0.8541  &  1.0000 \\
Their proportions                  & 0.9918 &   0.9814  &  0.9669  &  0.9481  &  0.9249  &  0.8972  &  0.8648  &  0.8276  &  0.7854 \\
\hline
\end{tabular}}
\caption{\label{InertialPropTable} Values of inertial proportion of a
  $17 \, \mathrm{cm}$ long waveguide with its comparison values. The
  quarter wavelength frequency of such tube is $500 \, \mathrm{Hz}$,
  corresponding to $F_1$.}
\end{table}

We conclude that between $100 \ldots 200 \, \mathrm{Hz}$, the acoustic
impedance of a $17 \, \mathrm{cm}$ long, uniform diameter tube (with
Dirichlet boundary condition at the far end opening) may be
approximated by the expression $Z_{ac}(2 \pi f i) \approx P_{iner}(f)
Z_0 i$, $Z_0 = \rho c / A$ without making error larger than $15$ \% in
impedance values. This serves as a ``rule of thumb'' for accuracy
whenever we consider the acoustic impedance of a waveguide more than
one octave below its lowest resonance (here $500 \, \mathrm{Hz}$).
Note that the impedance of an infinitely long transmission line is
purely resistive $\rho c/A$ whereas the inertial transfer function
$Z_{iner} = C_{iner} s$ is purely reactive.

\section{\label{WebsterSec} Inertial limit from Webster's equation}

 We concluded above that the inertance $C_{iner}$ of a constant
 diameter fluid column is given by $C_{iner} = \rho \ell /A$ where
 $\ell$ is the length and $A$ is the intersectional area. Because mass
 inertia is an additive quantity, the \emph{flow mechanical} inertance
 of the variable diameter waveguide is obtained from this by
 integrating the infinitesimal contributions $\rho \,dx/A(x)$; i.e.,
 \begin{equation} \label{WaveguideInertance}
   C_{iner} = \rho \int_0^\ell {\frac{dx}{A(x)}}.
 \end{equation}
 We proceed to show that the same expression for the \emph{acoustic}
 inertance can be concluded from the acoustic waveguide with a varying
 area function $A(x)$, $x \in [0, \ell]$. For a mathematical treatment
 of such waveguides through Webster's partial differential equation,
 see \cite{L-M:WECAD,A-L-M:AWGIDDS,L-M:PEEWEWP}.

 Since there is no possibility of expressing the impedance transfer
 function of such a waveguide in a closed form in the same manner as
 in Section~\ref{WaveGuideSubSec}, the argument must be carried out by
 an \emph{a priori} technique -- studying the partial differential
 equation rather than its solution. This is always much more
 difficult, and we only make the computations for the trivial
 termination transfer function $Z_L(s) = 0$.

 We begin by identifying the differential equations for the impedance
 transfer function of the waveguide using boundary and system nodes
 introduced in \cite{M-S-W:HTCCS,M-S:CBCS,M-S:IPCBCS}. Any internally
 well-posed boundary node $\Xi = \left (G, L, K \right )$ (such as the
 one coming from Webster's horn model) induces an infinite-dimensional
 system node $S = \SmallSysNode$. This system node gives rise to the
 dynamical system of type
 \begin{equation}\label{SystemNodeDynamics}
 \begin{aligned}
 \dot z(t) &= A_{-1}z(t) + Bu(t), \\
      y(t) &= Cz(t) + Du(t)\quad \text{ for } t \geq 0,\\
      z(0) &= 0.
 \end{aligned}
 \end{equation}
 If $\Xi$ is also impedance passive, then \emph{semigroup generator}
 $A := L\rst{\Null{G}}$ of $S$ is maximally dissipative, and the
 \emph{transfer function} of $S$ is defined by $\IOhat(s) := \CDop
 \sbm{ (s - A_{-1})^{-1}B \\ I}$ for all $s \in \cplus$.  The transfer
 function can always be expressed in terms of the original boundary
 node $\Xi$ as follows:

 \begin{proposition}
   \label{TransferFunctionProp}
   Let $\Xi = \left (G, L, K \right )$
   be a boundary node
   associated to the operator node $S = \SmallSysNode$ whose transfer
   function is $\IOhat(\cdot)$.
   \begin{enumerate}
   \item \label{TransferFunctionPropClaim1} 
     Then $y_s = \IOhat(s) u$ for $u \in \Uscr$,
     $s \in \rho(A)$, and $y_s \in \Yscr$ if and only if there exists a
     (unique) $z_s \in \Dom{\Xi}$ such that 
     \begin{equation} \label{TransferFunctionEq} \bbm{G \\ L \\
    K}z_s = \bbm{u \\ s z_s \\ y_s}.
     \end{equation}
   \item 
     \label{TransferFunctionPropClaim2} 
     Then $y_s = \IOhat'(s) u$ for $u \in \Uscr$, $s \in \rho(A)$, and
     $y_s \in \Yscr$ if and only if there exists a (unique) $z_s \in
     \Dom{\Xi}$ and $x_s \in \Null{G}$ such that
     \begin{equation} \label{TransferFunctionEq2} \bbm{G \\ L }z_s = \bbm{u \\ s z_s}
       \text{ and } \bbm{s - L \\ - K}x_s = \bbm{z_s \\ y_s}.
     \end{equation}
   \end{enumerate}
   In fact, $z_s \in \cap_{k \geq 1}{\Dom{L^k}}$.
 \end{proposition}
 \begin{proof}
   We use the following relations between the operators in $\Xi$ and
   $S$: $\CDop = \bbm{K & 0} | \Dom{S}$, $A = L | \Null{G}$, $A_{-1} =
   L - BG$ and $G(s - A_{-1})^{-1}B = I$.
   
   Claim~\eqref{TransferFunctionPropClaim1}: Now $y_s = \IOhat(s) u$ if
   and only if $y_s = \CDop \sbm{ (s - A_{-1})^{-1}B \\ I}u$ if and only
   if $y_s = Kz_s$ where $z_s = (s - A_{-1})^{-1}Bu$ if and only if
   $y_s = Kz_s$ where $(s - A_{-1})z_s = Bu_s$ if and only if $y_s =
   Kz_s$ where $(s - L)z_s = B(u - Gz_s)$.  But always $Gz_s = G(s -
   A_{-1})^{-1}Bu = u$, and hence $y_s = \IOhat(s) u$ is equivalent
   with the solvability of $z_s$ in \eqref{TransferFunctionEq}.
   Because $\Xi$ is a boundary node, $\sbm{G \\ s - L}$ is injective,
   and the solution $z_s$ of \eqref{TransferFunctionEq} is unique.
   
   Claim~\eqref{TransferFunctionPropClaim2}: This time $y_s =
   \IOhat'(s) u$ if and only if 
   \begin{equation*}
     y_s = - \CDop \bbm{(s - A)^{-1} (s -
       A_{-1})^{-1}B \\ I}u = K (s - A)^{-1} \cdot (s - A_{-1})^{-1}B u = K x_s
   \end{equation*}
   where $x_s = (s - A)^{-1} z_s$ and $z_s = (s - A_{-1})^{-1}B u$. By
   Claim~\eqref{TransferFunctionPropClaim1}, the vector $z_s$ is the
   unique solution of the first equation in
   \eqref{TransferFunctionEq2}. Moreover, we have $x_s \in \Null{G} =
   \Dom{A}$, and thus $(s - L) x_s = (s - A) x_s = z_2$ being
   equivalent with the second equation in \eqref{TransferFunctionEq2}.
   
   It follows from \eqref{TransferFunctionEq} that $z_s = (s -
   A_{-1})^{-1}Bu$ satisfies $Lz_s = sz_s$, and hence $z_s \in \cap_{k
     \geq 1}{\Dom{L^k}}$.
 \end{proof}

It remains to apply Proposition~\ref{TransferFunctionProp} to the
Webster's waveguide model given by
\begin{equation} \label{WebstersEqBnrCtrl} 
  \begin{cases}
    & \psi_{tt} = \frac{c^{2}}{A(s)} \frac{\partial}{\partial s}
     \left ( A(s) \frac{\partial \psi }{\partial s} \right )  \text{ for }  s \in (0,\ell) \text{ and } t \in \rplus, \\
    & - A(0) \psi_s(0,t)   =  i_0(t) \quad \text{ for } t \in \rplus, \\
    &  \psi_t(\ell,t) = \psi(\ell,t) = 0 \quad  \text{ for } t \in \rplus.
\end{cases}
\end{equation}
together with the observed signal
\begin{equation}\label{WebstersEqBnrCtrlObs}
  p_0(t) = \rho \psi_t(0,t)  \text{ for } t \in \rplus.
\end{equation}
Using the operator
\begin{equation*} \label{WebsterOpDef}
  W := \frac{1}{A(x)} \frac{\partial}{\partial x} \left ( A(x)
  \frac{\partial}{\partial x} \right ),
\end{equation*}
the first of the equations in
\eqref{WebstersEqBnrCtrl} can be cast into first order form by using
the rule
\begin{equation} \label{WebsterMainOpDef}
   \psi_{tt} = c^{2} W \psi \quad \hat{=} \quad \frac{d}{dt} \bbm{\psi \\ \pi} 
   = L  \bbm{\psi \\ \pi} \text{ where } L := \bbm{0 & \rho^{-1} \\ \rho c^2 W & 0}.
\end{equation}
Note that the rule implies $\pi = \rho \psi_t$.  We have $L :\Zscr \to
\Xscr$ where the two Hilbert spaces are given by
\begin{equation} \label{WebstersHilbertSpaces}
  \Zscr := H\sp{2}_{\{\ell \}}(0,\ell) \times H\sp{1}_{\{\ell \}}(0,\ell), \quad \Xscr :=
  H\sp{1}_{\{\ell \}}(0,\ell) \times L\sp{2}(0,\ell)
\end{equation}
where the subindex $\{\ell \}$ denotes the Dirichlet boundary
condition at $\ell$.
The endpoint control and observation operators $G:\Zscr \to \C$ and
$K:\Zscr \to \C$ are defined by
\begin{equation*}
   G \sbm{w_1 \\ w_2} := - A(0) w_1'(0)  \quad \text{ and } \quad
   K \sbm{w_1 \\ w_2} :=  w_2(0) .
\end{equation*}
Now the Webster's horn model
\eqref{WebstersEqBnrCtrl}--\eqref{WebstersEqBnrCtrlObs} for the state
$z(t) = \sbm{\psi(t) \\ \pi(t)}$ takes the form
\begin{align} \label{WebstersDiffEqNoF}
\begin{cases}
   \frac{d}{dt}  \sbm{\psi(t) \\ \pi(t)} &  =  L \sbm{\psi(t)
    \\ \pi(t)},   \\
    G \, \sbm{\psi(t) \\ \pi(t)}  & = i_0(t),
\end{cases}
\end{align}
and 
\begin{equation}
  \label{WebstersDiffEqNoFObs}
  p_0(t) = K \sbm{\psi(t) \\ \pi(t)}
\end{equation}
for all $t \in \rpluscl$. We have now constructed a (impedance
passive, internally well-posed) boundary node $\Xi = \left ( G, L, K
\right )$ whose transfer function $\IOhat(s) = Z_{ac}(s)$ is the
impedance of the acoustic waveguide when the far end has been
terminated to a Dirichlet boundary condition.  We now wish to compute
$Z_{ac}(0)$ and $Z_{ac}'(0)$, leading to the following result.

\begin{theorem}
  The trivial termination transfer function $Z_L(s) = 0$ is inertially
  feasible for the acoustic waveguide described by
  \eqref{WebstersEqBnrCtrl}--\eqref{WebstersEqBnrCtrlObs}, and the
  inertial factor satisfies $r_{iner} = 1$.
\end{theorem}
In other words, for Dirichlet terminated general acoustic waveguide,
the flow mechanical inertance and the acoustic inertance coincide.
Thus, Eq.~\eqref{InertanceExpression} holds for general acoustic
waveguides in the special case $Z_L(s) = 0$.
\begin{proof}
  We must show that $Z_{ac}(0) = 0$ and $Z_{ac}'(0) = C_{iner}$ where
  $C_{iner}$ is given by \eqref{WaveguideInertance}. The first step in
  this direction is to solve $\sbm{G \\ L }z_0 = \sbm{1 \\ 0 \\ 0}$.
  Writing $z_0 = \sbm{w_1 \\ w_2}$, we get the differential equation
  \begin{equation*}
    \begin{cases}
      & - A(0) w_1'(0) = 1 \\
      & \rho^{-1} w_2 = 0  \\
      &  \frac{1}{A(x)} \frac{\partial}{\partial s} \left ( A(x) 
  \frac{\partial w_1}{\partial x} \right ) = 0. 
    \end{cases}
  \end{equation*}
Thus $w_2(x) = 0$ and $w_1(x) = C_1 \int_{0}^x {\frac{d r}{A(r)}} +
C_2$ for $x \in [0, \ell]$. The boundary condition $w_1'(0) = -1/A(0)$
$C_1 = -1$, and $w_1(\ell) = 0$ implies $C_2 = \int_{0}^\ell {\frac{d
    r}{A(r)}}$. Thus, we have $z_0(x) = \sbm{\int_{x}^\ell {\frac{d
      r}{A(r)}} \\ 0}$. Trivially $K Z_0 = 0$, and $Z_{ac}(0) = 0$
follows from Claim~\eqref{TransferFunctionPropClaim1} of
Proposition~\ref{TransferFunctionProp}.

Writing now $x_0 = \sbm{w_1 \\ w_2}$, equation $\sbm{-L \\ -K }x_0 =
z_0 = \sbm{\int_{x}^\ell {\frac{d r}{A(r)}} \\ 0 \\ y}$ implies
$w_2(x) = - \rho \int_{x}^\ell {\frac{d r}{A(r)}}$. By
Claim~\eqref{TransferFunctionPropClaim2} of
Proposition~\ref{TransferFunctionProp} we get $Z'_{ac}(0) = y = - K
w_2 = \rho \int_{0}^\ell {\frac{d r}{A(r)}} = C_{iner}$. This
completes the proof.
\end{proof}

It remains an open question whether Eq.~\eqref{InertanceExpression}
can be generalised to general acoustic waveguides for any termination
impedance satisfying $Z_L(0) = 0$ and $Z_L'(0) \in \R$ . An educated
guess is that this can be done using a same kind but a more
complicated form of reasoning as presented in this section.


\begin{thebibliography}{10}

\bibitem{A-L-M:AWGIDDS}
A.~Aalto, T.~Lukkari, and J.~Malinen.
\newblock Acoustic wave guides as infinite-dimensional dynamical systems.
\newblock {\em ESAIM: Control, Optimisation and Calculus of Variations},
  21(2):324--347, 2015.
\newblock Published online: 17 October 2014.

\bibitem{A-A-M-M-V:MLBVFVTOPI}
A.~Aalto, T.~Murtola, J.~Malinen, D.~Aalto, and M.~Vainio.
\newblock Modal locking between vocal fold and vocal tract oscillations:
  Simulations in time domain.
\newblock Technical report, arXiv:1506.01395, 2017.

\bibitem{L-M:WECAD}
T.~Lukkari and J.~Malinen.
\newblock {W}ebster's equation with curvature and dissipation.
\newblock Technical report, arXiv:1204.4075, 27 pp. + 3 pp. appendix, 2013.
\newblock Submitted, revised.

\bibitem{L-M:PEEWEWP}
T.~Lukkari and J.~Malinen.
\newblock A posteriori error estimates for {W}ebster's equation in wave
  propagation.
\newblock {\em Journal of Mathematical Analysis and Applications},
  427(2):941--961, 2015.

\bibitem{M-S:CBCS}
J.~Malinen and O.~Staffans.
\newblock Conservative boundary control systems.
\newblock {\em Journal of Differential Equations}, 231(1):290--312, 2006.

\bibitem{M-S:IPCBCS}
J.~Malinen and O.~Staffans.
\newblock Impedance passive and conservative boundary control systems.
\newblock {\em Complex Analysis and Operator Theory}, 2(1):279--300, 2007.

\bibitem{M-S-W:HTCCS}
J.~Malinen, O.~Staffans, and G.~Weiss.
\newblock When is a linear system conservative?
\newblock {\em Quarterly of Applied Mathematics}, 64(1):61--91, 2006.

\bibitem{M-A-M-A-V:MLBVFOVTA}
T.~Murtola, A.~Aalto, J.~Malinen, D.~Aalto, and M.~Vainio.
\newblock Modal locking between vocal fold oscillations and vocal tract
  acoustics.
\newblock {\em Acta Acustica united with Acustica}, 104(2):323--337, 2018.

\bibitem{M-A-M-G:PPMGFUIFHSV}
T.~Murtola, P.~Alku, J.~Malinen, and A.~Geneid.
\newblock Parametrisation of a physical model for glottal flow using inverse
  filtering and high-speed videoendoscopy.
\newblock {\em Speech Communication}, 96:67--80, 2018.

\bibitem{M-M:WPPGNVTR}
T.~Murtola and J.~Malinen.
\newblock Waveform patterns in pitch glides near a vocal tract resonance.
\newblock In {\em Proceedings of INTERSPEECH 2017}, pages 3487--3491, 2017.

\bibitem{M-M:IMBGSVTPG}
T.~Murtola and J.~Malinen.
\newblock Interaction mechanisms between glottal source and vocal tract in
  pitch glides.
\newblock In {\em Proceedings of INTERSPEECH 2018}, pages 2987--2991, 2018.

\bibitem{S-M-M:CFSPURVTG}
L.~Schickhofer, J.~Malinen, and M.~Mihaescu.
\newblock Compressible flow simulations of voiced speech using rigid vocal
  tract geometries acquired by {MRI}.
\newblock {\em Journal of the Acoustical Society of America}, 145(4):38--47,
  2019.

\end{thebibliography}
\end{document}